\documentclass[twocolumn,superscriptaddress]{revtex4-2}
\pdfoutput=1

\usepackage[utf8]{inputenc}
\usepackage{comment}
\usepackage[english]{babel}
\usepackage{amsfonts, amsmath, amssymb, amsthm}
\usepackage{xcolor}
\usepackage[hidelinks]{hyperref}
\usepackage{graphicx}

\newtheorem{definition}{Definition}
\newtheorem{lemma}[definition]{Lemma}
\newtheorem{theorem}[definition]{Theorem}
\newtheorem{proposition}[definition]{Proposition}
\theoremstyle{remark}
\newtheorem{remark}[definition]{Remark}

\usepackage{algorithmicx}
\usepackage{algorithm}
\usepackage[noend]{algpseudocode}

\algnewcommand\algorithmicinput{\textbf{Input:}}
\algnewcommand\Input{\item[\algorithmicinput]}

\algnewcommand\algorithmicoutput{\textbf{Output:}}
\algnewcommand\Output{\item[\algorithmicoutput]}

\usepackage{braket}

	\renewcommand\bra[1]{{\langle{#1}|}}
	\renewcommand\ket[1]{{|{#1}\rangle}}

\usepackage{todonotes}

\def\C{\mathbb{C}}
\def\D{\mathcal{D}}
\DeclareMathOperator{\tr}{Tr}
\DeclareMathOperator{\cut}{cut}
\def\E{\mathbb{E}}
\newcommand{\rel}{\mathrm{relax}}
\newcommand{\sgn}{\mathrm{sign}}
\newcommand{\M}{\mathcal{M}}

\begin{document}

\title{Approximate Solutions of Combinatorial Problems via Quantum Relaxations}

\author{Bryce Fuller}
\affiliation{IBM Quantum, IBM T.J. Watson Research Center, Yorktown Heights, NY 10598}
\author{Charles Hadfield}
\thanks{Previous affiliation: IBM Quantum, IBM T. J. Watson Research Center, Yorktown Heights, NY 10598}
\email{charles.hadfield@gmail.com}
\noaffiliation
\author{Jennifer R. Glick}
\affiliation{IBM Quantum, IBM T.J. Watson Research Center, Yorktown Heights, NY 10598}
\author{Takashi Imamichi}
\affiliation{IBM Quantum, IBM Research -- Tokyo, 19-21 Nihonbashi Hakozaki-cho, Chuo-ku, Tokyo, 103-8510, Japan}
\author{Toshinari Itoko}
\affiliation{IBM Quantum, IBM Research -- Tokyo, 19-21 Nihonbashi Hakozaki-cho, Chuo-ku, Tokyo, 103-8510, Japan}
\author{Richard J. Thompson}
\affiliation{Integrated Vehicle Systems, Boeing Research \& Technology, Huntsville, AL 35824}
\author{Yang Jiao}
\affiliation{Integrated Vehicle Systems, Boeing Research \& Technology, Tukwila, WA 98108}
\author{Marna M. Kagele}
\affiliation{Tech Vis and Integration, Global Technology, Boeing Research \& Technology, Tukwila, WA 98108}
\author{Adriana W. Blom-Schieber}
\affiliation{Product Development - Structures, Boeing Commercial Aircraft, Everett, WA 98204}
\author{Rudy Raymond}
\email{rudyhar@jp.ibm.com}
\affiliation{IBM Quantum, IBM Research -- Tokyo, 19-21 Nihonbashi Hakozaki-cho, Chuo-ku, Tokyo, 103-8510, Japan}
\author{Antonio Mezzacapo}
\email{mezzacapo@ibm.com}
\affiliation{IBM Quantum, IBM T.J. Watson Research Center, Yorktown Heights, NY 10598}
\begin{abstract}
Combinatorial problems are formulated to find optimal designs within a fixed set of constraints. They are commonly found across diverse engineering and scientific domains. Understanding how to best use quantum computers for combinatorial optimization is to date an open problem. Here we propose new methods for producing approximate solutions for the maximum cut problem and its weighted version, which are based on relaxations to local quantum Hamiltonians.
These relaxations are defined through commutative maps, which in turn are constructed borrowing ideas from quantum random access codes. We establish relations between the spectra of the relaxed Hamiltonians and optimal cuts of the original problems, via two quantum rounding protocols. The first one is based on projections to random magic states. It produces average cuts that approximate the optimal one by a factor of least $0.555$ or $0.625$, depending on the relaxation chosen, if given access to a quantum state with energy between the optimal classical cut and the maximal relaxed energy. The second rounding protocol is deterministic and it is based on estimation of Pauli observables.
The proposed quantum relaxations inherit memory compression from quantum random access codes, which allowed us to test the performances of the methods presented for 3-regular random graphs and a design problem motivated by industry for sizes up to 40 nodes, on superconducting quantum processors.
\end{abstract}


\maketitle

The idea that quantum computers can be used to generate approximate solutions for NP-hard combinatorial problems was proposed over two decades ago, using quantum adiabatic eigenstate evolution~\cite{farhi2000adiabatic}. It was then generalized to a quantum approximate optimization algorithm (QAOA), based on a variational optimization of quantum parameters~\cite{farhi2014qaoa,farhi2015quantum}.
The great interest in harnessing quantum advantage for classical combinatorial problems compelled many to perform extensive studies of the performance of the algorithm~\cite{zhou2020quantum,egger2020warm,tate2020bridging,nannicini2019performance, mcclean2020low}. These quantum optimizations rely on a bijective mapping between the space of classical binary variables and logical basis states of a collection of qubits, as illustrated in the original formulation~\cite{farhi2014qaoa}. 
Within the QAOA framework, the cost function to be optimized on a quantum computer has a classical maximal eigenstate, in the sense that superposition and entanglement are not strictly needed to prepare it. This is in contrast to quantum many-body Hamiltonians or quantum chemistry, whose extremal eigenstates often are very entangled. In the latter case, quantum computers have a natural memory advantage in storing the ground state. 

The algorithms presented here unify the two problems, in the sense that they produce approximate solutions of combinatorial problems searching for extremal eigenstates states of local quantum Hamiltonians. These local quantum Hamiltonians are relaxations of the original combinatorial problems; for each element in the image of a combinatorial cost function it is possible to construct a quantum state with the same Hamiltonian expectation value. 

\begin{figure*}
    \centering
    \includegraphics[width=0.95\linewidth]{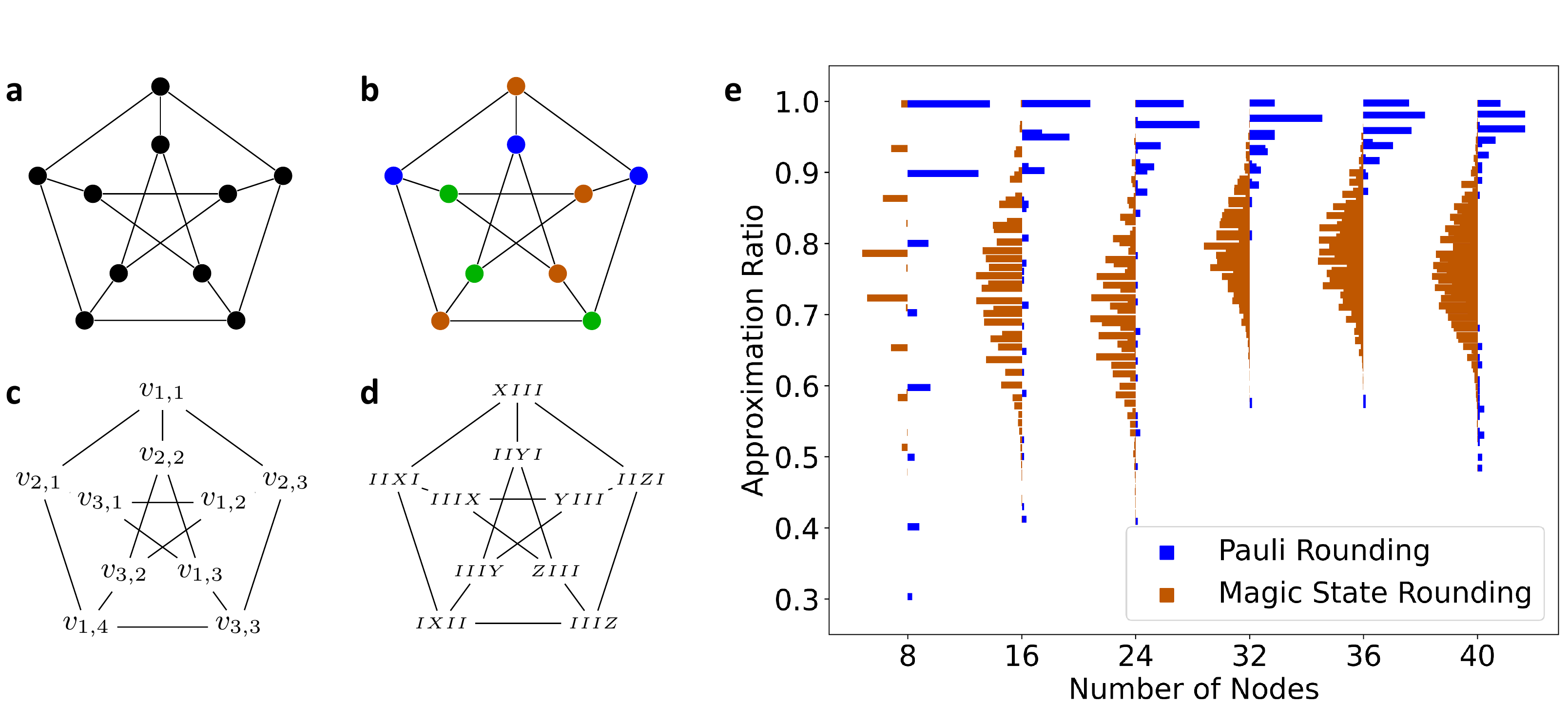}
    \caption{{\bf Hamiltonian Relaxation.}  The Petersen graph and the steps used to build the Hamiltonian relaxation: \textbf{(a)} A degree-3 graph with 10 vertices; \textbf{(b)} An LDF coloring of the graph; \textbf{(c)} Assignment of indices to colored vertices; \textbf{(d)} Assignment of Pauli operators to vertices. \textbf{(e)} Rounding of relaxed ground states for MaxCut problems on 100 3-regular random graphs for each size of 8, 16, 24, 32, 36, and 40 nodes. We show the approximation ratio obtained for each rounding.}
    \label{Figure1}
\end{figure*}

Once a relaxed Hamiltonian is constructed, a candidate ground state can be obtained with variational eigensolvers~\cite{peruzzo2014variational}, quantum phase estimation~\cite{kitaev1995quantum}, approximation algorithms for quantum Hamiltonians~\cite{anshu2021improved}, or quantum imaginary time evolution~\cite{motta2020determining}. While we use the variational approach here, we note that other ground state preparation procedures could lead to improved performances on fault-tolerant devices.  We focus in the following on the maximum cut (MaxCut) problem and weighted MaxCut, which are NP-hard problems. Although we particularize our analysis to MaxCut, we believe that the ideas presented here can be readily generalized to any combinatorial problem. Once a quantum relaxed state is prepared, we round its relaxed energy to an admissible cut value via quantum measurements, following in spirit the well-known roundings used in classical algorithms, such as the semidefinite programming relaxations introduced by Goemans and Williamson~\cite{gw1995}. 

We start by recalling the the MaxCut problem. Consider a graph $G=(V,E)$ with $|V|$ vertices, labelled $\{v_i\}$, and $|E|$ edges, labelled $\{e_{i,j}\}$, ${i,j}\in[|V|]$. The degree of the graph is $\deg(G)$. We let MaxCut be the problem
\begin{equation}\label{eqn:cut}
    \max_{m\in\{-1,1\}^{|V|}}
    \cut (m)
    =
    \max_{m\in\{-1,1\}^{|V|}}\sum_{e_{i,j}\in E} \frac12 (1 - m_i m_j)
\end{equation}
We let $m^*$ denote a configuration that maximizes $\cut (m)$. Then
our quantum Hamiltonian relaxation of the MaxCut problem in Eq.~(\ref{eqn:cut}) can be formulated through the following proposition
\begin{proposition}
Given a graph $G=(V,E)$, there exists an embedding
\begin{equation}
    F: \{-1, 1\}^{|V|} \to \D(\C^{2^n})
\end{equation}
and a Hamiltonian $H$ on $n$ qubits that, in expectation, commutes with the MaxCut function $\cut$. That is, for all $m \in \{-1,1\}^{|V|}$, we have
\begin{equation}\label{eqn:HFcut}
    \tr\left( H \cdot F (m) \right) = \cut(m).
\end{equation}
The Hamiltonian is composed of (up to a multiple of the identity operator) a linear combination of Pauli operators of weight 2.
Moreover, the construction of $H$ and $F$ may be performed in time complexity $\tilde O (|V|)$.
\end{proposition}
To construct $H$ we first briefly summarize quantum random access codes~\cite{ambainis1999dense,ANTV02}.
A quantum random access code (QRAC) allows more than one binary variable to be encoded per qubit at the expense of retrieval of each binary-variable value becoming probabilistic. The optimization of these two parameters is constrained by Nayak's bound~\cite{Nayak99} that extends Holevo's bound~\cite{Holevo1973}.
We consider the $(3,1,p)$ QRAC that encodes three binary variables $m = \{ m_i \}_{i\in[3]}$ into one qubit $\D(\C^2)$
\begin{equation}\label{eqn:31encoding}
    f(m)
    =
    \tfrac12
                \left(
                    I + \tfrac1{\sqrt3}  \left(  
                                            m_1 X + m_2 Y + m_3 Z
                                        \right)
                \right)
\end{equation}
such that, with probability $p=\frac12+\frac1{2\sqrt3}$, a given binary variable may be recovered upon measuring in the corresponding Pauli basis~\cite{Hayashi2006}. For example, $\tr(X \cdot f(m)) = \frac1{\sqrt3} m_1$, hence, if we write $\pi_X^+$ to be projection onto the $+1$ eigenspace of $X$, then $\tr(\pi_X^+ \cdot f(+1, m_2, m_3)) = \frac12+\frac1{2\sqrt3}$. Here we are not interested in this probabilistic retrieval of classical information in the original scope of QRACs, but rather to use them as tools to build our relaxations.

The next step is a coloring of the graph $G$. 
The MaxCut graph $G$ can be colored with the large-degree-first (LDF) method~\cite{WelshPowell67}. LDF finds a coloring $\{V_c\}_{c\in[C]}$ of the vertices $V$ into $C$ colors such that if $e_{i,j}\in E$ then the vertices $v_i, v_j$ are associated with different partitions $V_{c(i)}, V_{c(j)}$. In a compact form: $e_{i,j}\in E \implies c(i)\neq c(j)$.
LDF has time-complexity $O(|V|\log(|V|) + \deg(G)|V|)$

and performs such that the total number of colors satisfies $C\le \deg(G)+1$, see Fig.~\ref{Figure1}a and~\ref{Figure1}b for an example of coloring on the Petersen graph. This coloring operation dictates the runtime of the embedding in Proposition 1.
We then associate to each color $n_c$ qubits such that 

$n_c = \left\lceil |V_c| / 3 \right\rceil$ and, to each vertex of a given color, we associate a weight-1 Pauli operator that is supported on the respective $n_c$ qubits. The total number of qubits used here is reduced up to a factor 1/3 with respect to the standard QAOA formulation. The memory saving aspect alone is shared with previous compact encodings~\cite{tan2020qubit,patti2021nonlinear}. 
Note that up to two of the possible $3n_c$ Pauli operators per color may not be assigned. We build the embedding with the $(3,1,p)$ QRAC of Eq.(\ref{eqn:31encoding}) in mind, but other encodings can be used. For example, the case of the $(2,1,0.85)$ QRAC is derived in the Supplementary Information.

Once every set of $n_c$ qubits is associated to a set of vertices $V_c$, we can construct the relaxed Hamiltonian.  
To each edge $e$ with associated vertices $v_{c,i},v_{c',i'}$ (note that $c\neq c'$), we declare the weight-2 Pauli operator $O_e = P_{c,i}P_{c',i'}$, where we write $P_{c,i}$ for the Pauli operator associated with a vertex $v_{c,i}$, where $i\in[|V_c|]$. 
See Figure~\ref{Figure1}d.
We are now in the position to define the following Hamiltonian:
\begin{equation}\label{eqn:H}
    H
    =
    \sum_{e\in E}
    \frac12
    \left(
        I 
        - 3\cdot
        O_e
    \right),
\end{equation}
which acts on $\sum_c n_c$ qubits.

It remains to make explicit the embedding $F$. Consider a vertex $v_{c,i}$ and the associated variable $m_{c,i}$. There are at most two other vertices $v_{c,j}$, $v_{c,k}$ (with associated variables $m_{c,j},m_{c,k}$) such that the Pauli operators $P_{c,i},P_{c,j},P_{c,k}$ are all supported on the same qubit. Index this qubit $q$. (If less than two such vertices exist, because $3\nmid|V_c|$, then introduce dummy variables $m_{c,j}, m_{c,k}$ as required, and fix their values to be $+1$.) The values of the variables $m_{c,i},m_{c,j},m_{c,k}$ are then encoded into the qubit indexed $q$ using the function $f$ from Eq.~\eqref{eqn:31encoding}. Therefore 
\begin{equation}\label{eqn:f_for_F}
    \tr( P_{c,i} \cdot f(m_{c,i},m_{c,j},m_{c,k}))
    =
    \tfrac{ 1 }{ \sqrt3 }
    m_{c,i}
\end{equation}
This procedure, applied to all vertices, provides the construction of the embedding $F$.

We obtain the commutation of Eq.~\eqref{eqn:HFcut} by combining Eq.~\eqref{eqn:cut} and Eq.~\eqref{eqn:H} with the observation that, for an edge $e$ with vertices $v_{c,i},v_{c',i'}$, Eq.~\eqref{eqn:f_for_F} implies
\begin{equation}
    \tr( O_e \cdot F(m) )
    =
    \tfrac13
    m_{c,i} \cdot m_{c', i'}.
\end{equation}
This can be used in conjunction with Eq.~(\ref{eqn:H}) to retrieve the commutative map Eq.~(\ref{eqn:HFcut}).

\begin{figure*}
    \includegraphics[width=0.95\linewidth]{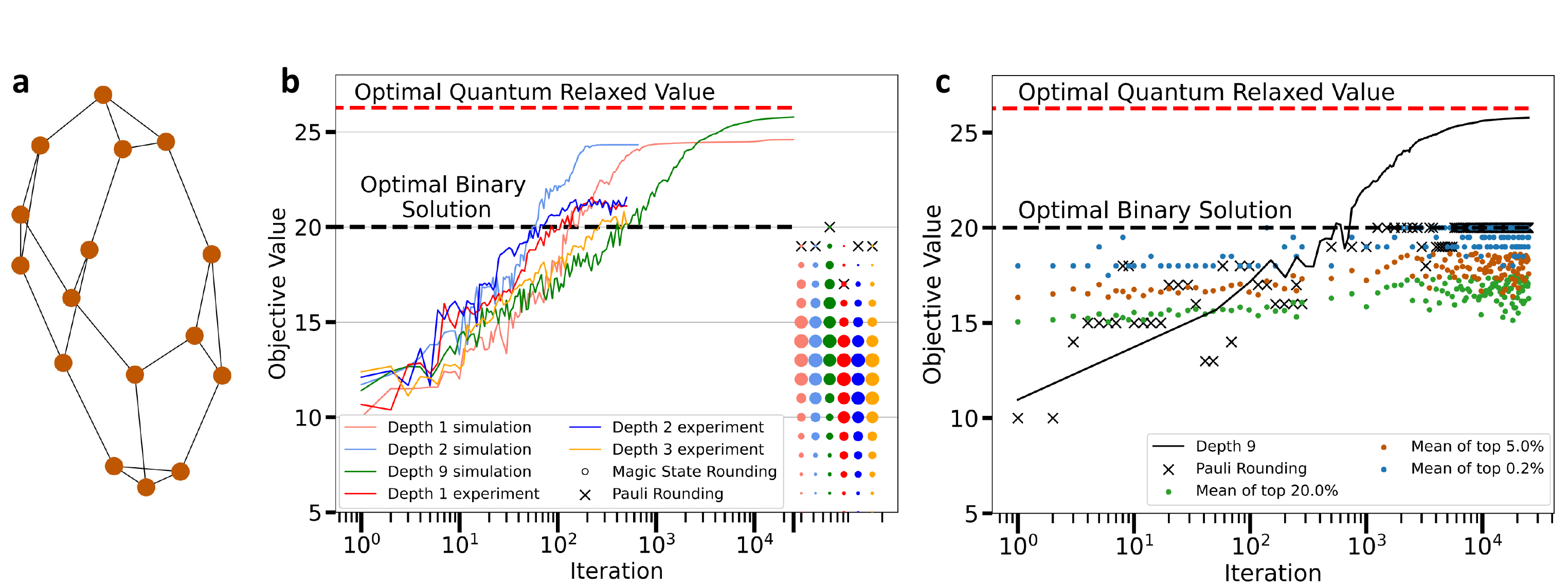}
    \caption{\textbf{Relaxed optimization and rounding.} \textbf{(a)} The 16-node 3-regular MaxCut problem graph used to benchmark our relaxed optimization. \textbf{(b)} Quantum variational search for a solution of the relaxed problem, with variational circuits of different depths. The optimal relaxed value (dashed red line) is larger than the optimal combinatorial solution (dashed black line). Variationally-optimized solutions are rounded to admissible values according to magic state and Pauli rounding.
    \textbf{(c)} Robustness of rounding versus optimizer iterations. The convergence of the top 20\%, 5\%, and 0.2\% rounded cuts for magic state rounding, and the cut value obtained with Pauli rounding, are shown as a function of iteration number.}
    \label{Figure2}
\end{figure*}

Once the a relaxed Hamiltonian $H$ is obtained, we use a quantum computer to find an approximation $\rho_\rel$ to the highest-energy state associated with $H$. Note that the highest-energy state is mapped to the ground state if $H\rightarrow -H$. While in our experiments we have focused on  a variational approach, we note the best approximation to the ground state can be obtained in other ways~\cite{kitaev1995quantum,anshu2021improved,farhi2000adiabatic,motta2020determining}. 
The state $\rho_\rel$ is not, in general, in the image of $F$. For almost all graphs, its energy is strictly greater than the energy of the embedding of the optimal configuration of binary variables, $F(m^*).$ We are therefore required to provide a procedure that produces a candidate configuration $m$. We accomplish this task in two ways, with two rounding methods that we call {\it magic state rounding} and {\it Pauli rounding}. 

For the first rounding method, consider the single-qubit magic states 
\begin{equation}\label{eqn:3-1-p-qrac}
    \mu^\pm
    =
    \tfrac12
                \left(
                    I \pm \tfrac1{\sqrt3}  \left(  
                                            X + Y + Z
                                        \right)
                \right)
\end{equation} 
and set $\mu_1^\pm = X\mu^\pm X$, $\mu_2^\pm = Y\mu^\pm Y$, $\mu_3^\pm = Z\mu^\pm Z$, $\mu_4^\pm = \mu^\pm$. Measuring in the bases associated with $\mu_i^\pm$ for $i\in[4]$ corresponds to estimating the expectation value of $\mu_i^+ - \mu_i^-$.
We define magic state rounding as the procedure $\M$ that takes a single-qubit density $\rho$ and uniformly at random selects a measurement basis $\{\mu_i^+,\mu_i^-\}_{i\in[4]}$ in which to measure a quantum state $\rho$. This provides a state $\M(\rho)$ that, in expectation, is
\begin{equation}
\label{eq:depolarizing}
\begin{aligned}
    \E(\M(\rho))
    &=
    \E_{i\in[4]} \E_{s\in\{\pm \}} (\M(\rho))
    \\
    &=
    \frac14 \sum_{i\in[4]} \sum_{s\in\{\pm \}} \tr(\mu_i^s \rho) \mu_i^s
    \\
    &=
    \mathcal{E}_{1/3}(\rho)
\end{aligned}
\end{equation}
where $\mathcal{E}_\delta$ is the single-qubit depolarizing channel: $\mathcal{E}_\delta(I)=I$ and $\mathcal{E}_\delta(P)=\delta P$ for $P\in\{X,Y,Z\}$. 
On a quantum processor, which can natively only perform measurement in the $Z$-basis, measuring in magic bases requires only single-qubit unitaries. We provide a detailed description of those basis change unitaries in Supplementary Information~\ref{app:measurement}.
We apply $\M$ to all qubits independently and denote this quantum channel $\M^{\otimes n}$. Upon application of $\M^{\otimes n}$, the state of each qubit is one of $\mu_i^\pm$, hence the three binary variables associated with each qubit become fixed via Eq.~\eqref{eqn:31encoding}, and the associated classical decision variables are given via Eq.~(\ref{eqn:31encoding}). We have established a map
\begin{align}
    \M^{\otimes n} : \rho_\rel \mapsto F(m)
\end{align}
for some $m$.
Using the result in Eq.~(\ref{eq:depolarizing}), we now show a relaxation bound for the average cut obtained with the map $\M^{\otimes n}$. First, let $\gamma$ denote the approximation ratio for a fixed configuration $m$ (relative to any optimal configuration $m^*$). That is,
\begin{equation}\label{eqn:approx-ratio}
    \gamma
    =
    \frac{\cut(m)}{\cut(m^*)}
    =
    \frac{ \tr (H \cdot F(m) ) }
        { \tr(H \cdot F(m^*)) }.
\end{equation}
Let us refer to the procedure described above as magic state rounding. We are now in the position of presenting the following theorem
\begin{theorem}\label{thm:lower-bound}
    Given access to $\rho_\rel$, a quantum state with energy between $m^*$ and the maximal eigenstate of $H$, magic state rounding produces a configuration $m$ whose expected approximation ratio is at least $5/9$. That is,
\begin{equation}\label{eqn:magic-approx-ratio}
    \E(\gamma)
    =
    \frac{ \E\left( \tr (H \cdot \M^{\otimes n}(\rho_\rel) ) \right) }
        { \tr(H \cdot F(m^*)) }
    \ge \frac59.
\end{equation}
\end{theorem}
We give the proof of this result in the Supplementary Information. Furthermore, suppose $\rho_1,\rho_2$ are any states satisfying $\tr(H \rho_1)\ge \tr(H \rho_2)\ge |E|/2$. Then the proof of Theorem~\ref{thm:lower-bound} also indicates that the rounded solutions respect the inequality in expectation:
$
\E(\tr(H\cdot \M^{\otimes n}(\rho_1)))
\ge 
\E(\tr(H\cdot \M^{\otimes n}(\rho_2)))
$. This motivates looking for states of higher energy. The higher the enegy of the quantum state we round from, the better. In general, the rounded average expected value is linearly proportional to $\tr[H\rho]$; this second fact motivates looking for encodings whose relaxed spectrum is maximally separated from the optimal cut $m^*$.
We also recast the expected approximation ratio in the regime where the maximum possible cut value is small (MaxCutGain) in the Supplementary Information, where the guaranteed approximation ratios of classical approximation algorithms~\cite{gw1995} are not very useful. Finally, note that using a $(2,1,p)$ QRAC encoding, the expected approximation ratio $\E(\gamma)$ increases to $5/8=0.625$.

The second rounding procedure we propose is simpler than magic state rounding, and closer in spirit to the bit recovery used in quantum random access codes. It is based on the direct estimation of the Pauli operators $P_{c,i}$ associated to each vertex $v_{c,i}$. The graph variables are then chosen according to $m_{c,i} = \sgn(\tr (P_{c,i}\cdot \rho))$. If $\tr (P_{c,i}\cdot \rho)$ is exactly zero, we assign a value uniformly at random. We call this {\it Pauli rounding}. 

One important aspect to consider is the measurement cost of estimating $\tr (H\rho)$ if a candidate ground state is obtained with variational algorithms. The following lemma shows a dependence logarithmic in the system size $|E|$ of the error $\epsilon$ for an estimate of the approximation ratio $\gamma$ in Eq.~(\ref{eqn:approx-ratio}) - a multiplicative error. We use classical shadows via random Pauli measurements \cite{hkp20} whose variance we can bound as observed in the following:
\begin{lemma}
    The classical shadows estimator for $\tr(H \rho)$ achieves an estimate with at most $\varepsilon$ multiplicative error with success-rate $1-\delta$ provided the number of measurements $S$ satisfies $S > \frac{2\cdot 3^4}{\varepsilon^2} \log ( 2 |E| / \delta )$.
\end{lemma}
The lemma is proved in the Supplementary Information. Furthermore, numerical evidence suggests that the number of measurements required will be significantly reduced if we use optimized randomized estimators~\cite{hadfield20, huang2021efficient, hadfield2021adaptive, wu2021overlapped}. For completeness we also prove, in Supplementary Information~\ref{app:recovery_cut_given_rho}, that if a state $\rho$ is known to be in the image of $F$, then the estimation of $\tr(H \rho)$, with success rate $1-\delta$, requires a number of measurements $S$ satisfying
$
    S 
    >
    {2\cdot 3^{ 4 }}
    \cdot
    \log ( 2 |E| / \delta )
$.

\begin{figure*}
\includegraphics[width=1.0\linewidth]{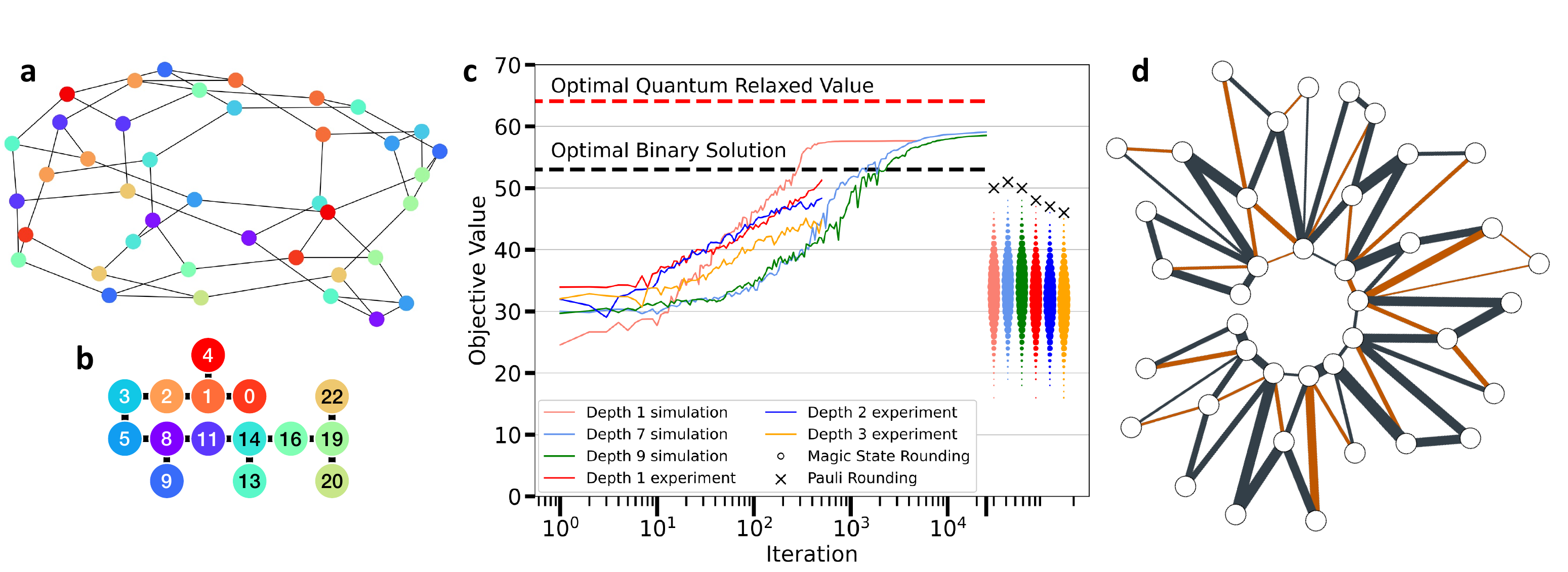}
\caption{\label{Figure3}  \textbf{Experiments for unweighted MaxCut on a 40-node 3-regular graph and weighted MaxCut on a 40-node planar graph.} \textbf{(a)} The 40-node unweighted 3-regular graph, the largest graph on which we tested the algorithm. \textbf{(b)} The qubits of the \emph{ibmq\_dublin} quantum processor used for the experiments. We map 40 vertices into 15 qubits. Vertices assigned to the same qubits are colored the same. \textbf{(c)} Performance of the relaxed optimization for the MaxCut on the 40-node 3-regular graph in \textbf{(a)}. We obtain approximation ratios of $\gamma = 0.962$ for theory simulation and $\gamma = 0.905$ for hardware experiment on the 16 qubits in \textbf{(b)}, using Pauli rounding.
\textbf{(d)} A weighted MaxCut problem on a planar graph relevant for industrial aerospace design, see Supplementary Information for details on the formulation of the problem. The thickness of the graph edges represents the weights of the MaxCut problem and the black (orange) edges represent satisfied (unsatisfied) constraints from a hardware experiment. The solution obtained has a cost of 617 for an optimal cost for this problem of 641.}
\end{figure*}

We now present numerical simulations and hardware experiments benchmarking the relaxation methods proposed. In order to test the performance of the two quantum rounding methods presented, we first perform numerical simulations on 100 random MaxCut instances on 3-regular graphs of increasing size, with $n = 8, 16, 24, 32, 36, 40$, see Fig.~\ref{Figure1}e.  We construct relaxed Hamiltonians using (3,1,p) QRAC embeddings for each MaxCut graph. On average, these relaxed Hamiltonians have a number of qubits reduced by a factor 2.6 with respect to the number of nodes in the original 3-regular graphs. We numerically find their ground states via exact diagonalization and perform magic state rounding 100 times, separately for each graph. For every call to magic state rounding, we compute the corresponding approximation ratio $\gamma$ as in Eq.~(\ref{eqn:approx-ratio}), using IBM ILOG CPLEX~\cite{cplex} to obtain the optimal cuts $m^*$. We show all the approximation ratios via histograms, observing that on average we indeed obtain cuts greater than $5/9$. We then benchmark Pauli rounding on the same problems, computing exactly for each ground state $\tr(P_{c,i}\rho)$ for every vertex. We observe that Pauli rounding obtains better approximation ratios on average. However, note that while Pauli rounding is a deterministic procedure, the stochastic component of magic state rounding could lead to better sampled solutions for some graphs.   

Next, we probe the robustness of the algorithm when used in conjuction with a variational eigensolver, used to find a candidate ground state $\rho_\rel$. We pick two 3-regular graphs from the random ensemble generated for Fig.~\ref{Figure2}e, of 16 and 40 nodes, with exact approximation ratio $\gamma = 1$ according to Pauli rounding. We start with the 16-node graph, represented in Fig~\ref{Figure2}a. We search for the ground state of its relaxed Hamiltonian both numerically and on the \emph{ibmq\_dublin} quantum processor. In both cases, we use a hardware-efficient ansatz~\cite{kandala2017hardware} with different depths, see Fig.~\ref{Figure2}b. The numerical experiments are conducted using the COBYLA optimizer for up to 25000 iterations, while the experiments on the superconducting processor were performed with an adaptive SPSA optimizer~\cite{eddins2021doubling}, using 500 iterations and 8192 measurements for each independent basis to estimate the relaxed Hamiltonian. For further details about the experiments, see Supplementary Information.

We observe in the simulations that the trial states reach relaxed values of the objective function---our Hamiltonian ``energy''---which are inaccessible to the combinatorial cost function. Once the state is optimized, we perform magic state rounding 1000 times to obtain admissible cut distributions, which are represented by the colored markers on the bottom right of Fig.~\ref{Figure2}b. We then apply Pauli rounding to the optimized states for the different depths, obtaining cut values indicated by the ``X'' markers. The state prepared with a trial circuit of depth 9 obtains an approximation ratio of 1, even if its optimized objective value did not match exactly the optimal relaxed value, hinting at a robustness from the relaxed formulation of the problem. We then perform a variational optimization on the \emph{ibmq\_dublin} quantum processor, observing even on noisy hardware that relaxed optimized energies go above the optimal combinatorial solution. 

We carry out numerical simulations to show the robustness of the rounding schemes on this graph against an incomplete optimization search in Fig.~\ref{Figure2}c. The quantum state is now rounded as it goes through the variational optimization, observing that good convergence is achieved for both rounding schemes already around $\sim 1000$ iterations, without requiring an accurate ground state preparation, which happens at the end of the 25000 optimization steps.

We then test the proposed methods on the largest graph considered here, a 40-node graph represented in Fig.~\ref{Figure3}a. We map its 40 vertices into 15 qubits of the \emph{ibmq\_dublin} quantum processor, color coded in Fig.~\ref{Figure3}b. The performance of the algorithm is shown in Fig.~\ref{Figure3}c for a hardware-efficient ansatz of various depths. The quantum state is optimized by maximizing mean values for a relaxed Hamiltonian computed on the 15 qubits of Fig.~\ref{Figure3}b, then rounded using both magic state and Pauli rounding. Approximation ratios of $\gamma = 51/53\approx 0.9623$ and $\gamma = 0.905$ are obtained for theory and hardware runs, respectively. 

Finally, we consider a weighted MaxCut problem $    \max_{m\in\{-1,1\}^{|V|}}\sum_{e_{i,j}\in E} \frac12 w_{ij}(1 - m_i m_j)$ with edge weights $w_{ij}>0$ on a planar graph on 40 nodes, arising in the context of industrial aerospace design. We show that our algorithm can be readily extended to weighted MaxCut in the Supplementary Information, where we also provide more details about the formulation of the design problem. In Figure~\ref{Figure3}d, we report the best solution obtained in experiments, with Pauli rounding, by plotting the weighted graph, displaying the satisfied constraints in black and unsatisfied constraints in orange. As expected within the weighted MaxCut formulation, we observe that most high-weight constraints are satisfied. We experimentally obtain an approximate solution of 617, where the optimal objective value is 641. 

The theoretical and experimental results presented here show that it is possible to produce approximate solutions of MaxCut and weighted MaxCut by preparing on quantum computers approximate high energy states of quantum Hamiltonians, which are relaxations of the original combinatorial problems. The largest experiments conducted on a superconducting quantum processor have shown approximation ratios for 40-node graphs of 0.905 on unweighted 3-regular MaxCut and $617/641\approx0.9626$ on weighted planar MaxCut. Through a randomized quantum rounding procedure, we have established rigorous connections between the energy of a relaxed quantum state and the average value of its rounded MaxCut approximation. We have proved that rounding from higher relaxed energy states leads to better combinatorial solutions. This motivates further research to look into classes of graphs that show large separations between the highest energy states of the relaxed spectrum and the optimal classical cut values. Approximation algorithms for many-body quantum states~\cite{AnshuGossetMorenz2020} could be employed to prepare ground state approximations for the relaxed problems. While we have considered only embeddings based on 1-qubit quantum random access codes, other more sophisticated embeddings can be considered which could result in improved relaxations. Other rounding protocols for relaxed states can be investigated, which for example could take into account non-local qubit correlations.

\section*{Acknowledgements and Author contributions}

We acknowledge useful conversations with Ramis Movassagh, Giacomo Nannicini, Stefan W{\"o}rner.
C.H., A.M., R.R. designed the algorithms. R.R. conceived the idea of applying QRAC for optimization. C.H. derived the approximation bounds for magic state rounding and the estimation errors. B.F. performed the experiments on the IBM Dublin device.  A.B., B.F., C.H., T.IM., T.IT., Y.J, M.K., R.T. performed the numerical simulations on the 3-regular graphs and the design problem. B.F., A.M., R.T., Y.J, M.K., A.B., J.G. formulated the design problem as a weighted MaxCut. All authors contributed to the data analysis and writing of the manuscript.

\bibliographystyle{unsrt}
\bibliography{bib}

\onecolumngrid
\newpage

\section*{Supplementary Information}
\section{Algorithms}\label{app:algorithm}

Here we summarie the algorithms described in the main text.

\begin{algorithm}[H]\label{alg:vqrao}
\caption{Approximate Optimization of MaxCut using Quantum Relaxation and Rounding}
\label{alg:maxcut_with_magic}
\begin{algorithmic}[1]
\Input Graph $G=(V,E)$; Number of quantum measurements $S$; A quantum rounding subroutine $R$.
\Output Approximate solution $m\in\{-1, 1\}^{|V|}$ to MaxCut problem.
    
\State Color graph $G$ using Largest-Degree-First algorithm giving $C$ colors and partition $V=\cup_{c\in[C]}V_c$
    
\State Index vertices $v_{c,i}$ where $c\in[C]$ and $i\in[|V_c|]$
    
    \For{color $c\in[C]$}
        \State{Assign $n_c = \left\lceil |V_c| / 3 \right\rceil$ qubits to variables $\{ m_{c,i} \}_{i\in[|V_c|]}$}
    \EndFor
    
\State Define embedding 
$F : \{-1, 1\}^{|V|} \to \D(\C^{2^n})$ where $n=\sum_{c\in[C]}n_c$

\State Define 2-local Hamiltonian $ H = \sum_{(i,i')\in E} \frac12 \left(I - 3\cdot P_{c,i}P_{c',i'} \right)$
\State Define an assignment $M_{V \mapsto P}: v_{c,i} \mapsto P_{c,i}$ from vertices to single qubit Pauli Observables in $H$.

\State Obtain $\tilde{\rho}_\rel$, an approximation  to the ground state of $H$, prepared by the quantum circuit $\mathcal{U}_\rel$

\State Obtain the approximate solution $m = R(\mathcal{U}_\rel, S, M_{V \mapsto P})$

\Return $m$
\end{algorithmic}
\end{algorithm}

\begin{algorithm}[H]\label{alg:magic-rounding}
\caption{Magic State Rounding Subroutine}
\begin{algorithmic}[1]
\Input An oracle $\mathcal{U}_\rel$ which prepares $\tilde{\rho}_\rel$; Number of quantum rounding steps $S$.
\Output Approximate solution $m\in\{-1, 1\}^{|V|}$ to MaxCut problem.

\State Initialize approximate solution $m=(1,1,\dots,1)$
    
    \For{quantum rounding $s\in[S]$}
        \State Denote by $m'$ the configuration to be obtained
        \State Prepare $\tilde{\rho}_\rel$ on quantum processor
        \For{qubit $q\in[n]$}
            \State Uniformly at random pick $j\in[4]$
            \State Measure qubit $q$ in magic basis $\{\mu_j^+,\mu_j^-\}$
            \State Assign values in $m'$ associated with qubit $q$ according to measurement
        \EndFor
        \If{$\cut(m')>\cut(m)$}
            \State Replace $m$ with $m'$
        \EndIf
    \EndFor
    \Return $m$

\end{algorithmic}  
\end{algorithm}

\begin{algorithm}[H]\label{alg:paulic-rounding}
\caption{Pauli Rounding Subroutine}
\begin{algorithmic}[1]
\Input An oracle $\mathcal{U}_\rel$ which prepares $\tilde{\rho}_\rel$; Number of quantum rounding steps $S$. An assignment $v_{c,i} \mapsto P_{c,i}$ from vertices to Pauli Observables.
\Output Approximate solution $m\in\{-1, 1\}^{|V|}$ to MaxCut problem.

\State Initialize approximate solution $m=(1,1,\dots,1)$
\State Define $ S_q =  \min(3, \max(\Vert V\Vert))$ to be the maximum number of variables    
 assigned to any one qubit. 
    \For{$v_{c,i} \in H$}
        \State Prepare $\tilde{\rho}_\rel$ on quantum processor
        \State Using $S_q$ measurements, estimate $\tr(H \cdot P_{c,i})$
        \If{$\tr(H \cdot P_{c,i})$ = 0}:
            \State Assign the value in $m$ associated with vertex $v_{c,i}$ uniformly at random
        \Else 
            \State Assign the value in $m$ associated with vertex $v_{c,i}$ according to $\text{sign}(\tr(H \cdot P_{c,i}))$
            \EndIf
        \EndFor
    \Return $m$

\end{algorithmic}  
\end{algorithm}


\newpage
\section{A relaxation bound}\label{app:relaxation}

In this section we provide a relaxation bound on the energy obtained by rounding quantum states with magic state rounding. Let $\rho_\rel$ be a quantum state with energy between $m^*$ and the maximal eigenvalue of the Hamiltonian $H$ and recall that $F(m^*)$ is a maximum-energy state within the image of $F$. We provide a probabilistic algorithm which takes $\rho_\rel$ and provides a state $F(m)=\M^{\otimes n}(\rho_\textrm{relax})$ such that, in expectation, $\E(\tr(H\cdot F(m))) \ge \frac59 \tr(H\cdot F(m^*))$. 
The idea is an application of Lieb's theorem \cite{lieb73} upon modifying an alternative proof of said theorem from \cite{bgkt2019}.

Consider the magic states $\mu^\pm = \frac12(I \pm \frac1{\sqrt3}(X+Y+Z) )$ on a single qubit and set
\begin{equation}\label{eqn:magic-bases}
    \mu_1^\pm = X\mu^\pm X,
    \qquad
    \mu_2^\pm = Y\mu^\pm Y,
    \qquad
    \mu_3^\pm = Z\mu^\pm Z,
    \qquad
    \mu_4^\pm = \mu^\pm.
\end{equation}
Measuring in the bases associated with $\mu_i^\pm$ for $i\in[4]$ corresponds to estimating the expectation value of $\mu_i^+ - \mu_i^-$ respectively.
Let $\M$ be the procedure that takes a single-qubit density $\rho$ and uniformly at random selects a measurement basis $\{\mu_i^+,\mu_i^-\}_{i\in[4]}$ in which to measure $\rho$. This provides a state $\M(\rho)$. In expectation
\begin{align}
    \E(\M(\rho))
    =
    \E_{i\in[4]} \E_{s\in\{\pm \}} (\M(\rho))
    =
    \frac14 \sum_{i\in[4]} \sum_{s\in\{\pm \}} \tr(\mu_i^s \rho) \mu_i^s
    =
    \mathcal{E}_{1/3}(\rho)
\end{align}
where $\mathcal{E}_\delta$ is the single-qubit depolarizing channel: $\mathcal{E}_\delta(I)=I$ and $\mathcal{E}_\delta(P)=\delta P$ for $P\in\{X,Y,Z\}$. This channel, when applied to a quantum state supported on multiple qubits, is entanglement breaking for $\delta\le 1/3$.

We extend the previous paragraph to the multi-qubit setting by considering the map $\M^{\otimes n}$ so that for traceless single-qubit Pauli operators $P,Q$ supported on different qubits, and an $n$-qubit density $\rho$, we obtain
\begin{align}\label{eqn:one_over_nine}
    \E\left(
        \tr(PQ \cdot \M^{\otimes{n}} (\rho))
    \right)
    = 
    \tr(PQ \cdot \mathcal{E}_{1/3}^{\otimes{n}} (\rho))
    =
    \tr(\mathcal{E}_{1/3}^{\otimes{n}}(PQ) \cdot \rho)
    =
    \frac19 \tr(PQ \cdot \rho).
\end{align}
The second equality follows from self-adjointness of $\mathcal{E}_\delta^{\otimes{n}}$ with respect to the inner-product $\langle A,B\rangle = \tr(A\cdot B)$.

The procedure of applying $\M^{\otimes n}$ produces a state which is in the image of $F$. If we therefore start with $\rho_\rel$, then we will recover some (random-variable) encoded state $F(m) = \mathcal{M}^{\otimes n}(\rho_\rel)$. Note that $H-\frac{|E|}{2}I$ is a Hamiltonian of weight-2 Pauli operators and therefore, by linearity of Eq.~\eqref{eqn:one_over_nine},
\begin{equation}\label{eqn:one_over_nine_linearity}
    \E\left( \tr\left((H-\tfrac{|E|}{2}I)\cdot \M^{\otimes n}(\rho_\rel) \right) \right)
    =
    \frac19 \tr\left((H-\tfrac{|E|}{2}I)\cdot \rho_\rel \right)
    \ge
    \frac19 \tr\left((H-\tfrac{|E|}{2}I)\cdot F(m^*) \right).
\end{equation}
The final term is bounded, that is, $\tr\left((H-\tfrac{|E|}{2}I)\cdot F(m^*) \right)$ takes values in the interval $[0, \frac{|E|}{2}]$.
Shifting the energy back to include the identity term recovers a lower bound for the approximation ratio. Specifically, let $\gamma$ denote the approximation ratio for a configuration $m$ (relative to any optimal configuration $m^*$), that is $\gamma = \cut(m) / \cut(m^*)$. Then, in expectation, this procedure leads to the following bound for the approximation ratio:
\begin{align*}
    \E(\gamma)
    &=
    \frac{ \E\left( \tr (H \cdot \M^{\otimes n}(\rho_\rel) ) \right) }
        { \tr(H \cdot F(m^*)) }
    \\
    &=
    \frac1{ \tr(H \cdot F(m^*)) }
    \left(
        \frac{|E|}{2}
        +
        \E\left( \tr\left((H-\tfrac{|E|}{2}I)\cdot \M^{\otimes n}(\rho_\rel) \right) \right)
    \right)
    \\
    & \ge
    \frac1{ \tr(H \cdot F(m^*)) }
    \left(
        \frac{|E|}{2}
        +
        \frac19
        \tr\left((H-\tfrac{|E|}{2}I)\cdot F(m^*) \right)
    \right)
    \\
    &=
    \frac{
        \frac{|E|}{2}
        +
        \frac19
        \tr\left((H-\tfrac{|E|}{2}I)\cdot F(m^*) \right)
    }{
        \frac{|E|}{2}
        +
        \tr\left((H-\tfrac{|E|}{2}I)\cdot F(m^*) \right)
    }
    \\
    &\ge
    \min_{\alpha\in\left[0,\frac{|E|}{2}\right]}
    \left(
    \frac{
        \frac{|E|}{2}
        +
        \frac19
        \alpha
    }{
        \frac{|E|}{2}
        +
        \alpha
    }
    \right)
    \\
    &=
    \frac{5}{9}
\end{align*}
We have proved that the relaxation procedure, followed by rounding using magic-bases measurements, provides a candidate encoded state which, in expectation, has an approximation ratio of at least $5/9$.

\begin{remark}
Consider two states $\rho_1,\rho_2$ which satisfy $\tr(H \rho_1)\ge \tr(H \rho_2)\ge |E|/2$. Then the equality in Eq.~\eqref{eqn:one_over_nine_linearity} reads
\begin{equation}
    \E\left( \tr\left((H-\tfrac{|E|}{2}I)\cdot \M^{\otimes n}(\rho_i) \right) \right)
    =
    \frac19 \tr\left((H-\tfrac{|E|}{2}I)\cdot \rho_i \right)
\end{equation}
for $i\in\{1,2\}$ and the energies in the preceding display are both positive since we have demanded $\tr(H \rho_i)\ge |E|/2$. Adding back the identity term and using the inequality relating the energies of $\rho_1$ and $\rho_2$ implies
\begin{equation}
    \E\left( \tr\left( H \cdot \M^{\otimes n}(\rho_1) \right) \right)
    \ge
    \E\left( \tr\left( H \cdot \M^{\otimes n}(\rho_2) \right) \right).
\end{equation}
Of course, this chain of reasoning can also be turned into a statement about the respective approximation ratios obtained from $\rho_1$ and $\rho_2$.
\end{remark}

\section{A relaxation bound for graphs with small MaxCut value}\label{app:relaxation_small_expansion}

In this section, we obtain a formula for the approximation ratio in the regime where the MaxCut value is small. 
In this regime, any bound of the form as that given in Theorem~\ref{thm:lower-bound}, including the lower bound obtained by Goemans and Williamson~\cite{gw1995}, becomes trivial. This is because for small MaxCut, such a bound will be worse than $|E|/2$; the value obtained, in expectation, by random assignment of the values of the binary variables.

Consider a graph $G=(V,E)$. Let us introduce the parameter $0<\varepsilon^*\le 1/2$ such that
\begin{equation}\label{eqn:gain-cut}
    \cut(m^*) = \left(\frac12 + \varepsilon^*\right) |E|.
\end{equation}
The variable $\varepsilon^*$ is called the gain, and the problem of evaluating $\varepsilon^*$ is called MaxCutGain.
We shall consider the regime $\varepsilon^*\ll 1/2$. Charikar and Wirth \cite{charikar04} give an efficient approximate classical algorithm for this problem whose approximation for the gain scales $\Omega(\varepsilon^*/\log(1/\varepsilon^*))$. Assuming the Unique Games Conjecture, this is optimal \cite{khot09}. Trevisan~\cite{TrevisanCut2009} shows a spectral partitioning algorithm for MaxCut whose approximation ratio is 0.531 which can also be used in this regime in addition to the regime $\varepsilon^* \approx 1/2$. There are only a few such algorithms that have non-trivial guarantees of MaxCut's approximation ratio.

Set $\varepsilon_\rel>0$ through the equation
$\tr(H\cdot \rho_\rel) = (\frac12+\varepsilon_\rel)|E|$.
An application of the triangle inequality establishes the upper bound $\varepsilon_\rel\le 3/2$. 
Our quantum rounding procedure implies 
\begin{equation}\label{eqn:approx-equality}
    \E\left( \tr (H \cdot \M^{\otimes n}(\rho_\rel) ) \right)
    = 
    \left(
    \frac12 + \frac19\varepsilon_\rel
    \right)
    |E|.
\end{equation}
Eq.~\eqref{eqn:gain-cut} and Eq.~\eqref{eqn:approx-equality} are effectively giving Lieb's theorem: $\frac19\varepsilon_\rel\le \varepsilon^*$. However in this context, the conclusion is that our algorithm, under the assumption of access to $\rho_\rel$, provides a candidate value for MaxCutGain which scales linearly without the logarithmic cost $1/\log(1/\varepsilon^*)$.
Of course, the Unique Games Conjecture does not hold when one considers quantum computational models \cite{kempe08ugc}. 
Furthermore, note that in the context of MaxCutGain it could be advantageous even to perform quantum rounding upon classical approximations of $\rho_\rel$, which according to ~\cite{harrow17magic} can deliver energies of at least $|E|(1/2 + 1/48d)$ for $d$-regular graphs.

\section{Measurement in magic bases on a quantum processor}\label{app:measurement}

We would like to measure a single-qubit density $\rho\in\D(\C^2)$ in one of the four magic bases defined in Eq.~\eqref{eqn:magic-bases}. This section explains how to perform such a measurement when we are given a quantum processor which only allows measurement in the $Z$-basis.

Let $\pi_Z^\pm$ denote the projections onto the $\pm1$-eigenspaces of the Pauli operator $Z$. (These projections are written, in the familar braket formalism, as $\pi_Z^+=\ket{0}\bra{0}$ and $\pi_Z^-=\ket{1}\bra{1}$.)
Consider a pure state $\psi=\frac12(I + aX+bY+cZ)$ with $a^2+b^2+c^2=1$. This state can be obtained from the initial state $\pi_Z^+$ by observing 
$\psi = e^{isZ}e^{itX}\pi_Z^+ e^{-itX}e^{-sZ}$ 
with $\cos^2t = \frac12(1+c)$ and $\sin(2s)=a/\sqrt{a^2+b^2}$.
Therefore, in order to rotate between the state $\pi_Z^+$ and the magic state $\mu_4^+$, rotations must be performed such that
\begin{equation}\label{eqn:st}
    \cos^2(t) = \frac12\left(1+\frac1{\sqrt3} \right),
    \qquad
    \sin(2s) = \frac1{\sqrt2}.
\end{equation}

Let us now see how to measure a density $\rho$ in any of the four magic bases:
\begin{itemize}
    \item Consider the simplest magic basis $\{\mu_4^+, \mu_4^-\}$. First, perform the necessary rotations so that $\mu_4^+$ is rotated into the state $\pi_Z^+$. That is, perform $e^{-isZ}$ then perform $e^{-itX}$ where $s,t$ are given in Eq.~\eqref{eqn:st}. (This rotates $\mu_4^-$ into $\pi_Z^-$.) Second, measure the state in the $Z$-basis which returns either $\pi_Z^+$ or $\pi_Z^-$. This is the end of the procedure since we do not continue to use the quantum processor.
    \item Consider now a general magic basis $\{\mu_i^+,\mu_i^-\}$ where $i\in[4]$. Set $P_i=X,Y,Z,I$ given respectively $i=1,2,3,4$. First, apply $P_i$. Second, apply the procedure outlined in the preceding bullet.
\end{itemize}

\section{Recovery of Hamiltonian expectation provided arbitrary state}\label{app:variance_classical_shadows}

Consider the task of estimating $\tr(H \rho)$ for a Hamiltonian associated with a graph $G=(V,E)$. We assume, for the moment, $\rho$ is an arbitrary state on the $n$-qubit quantum processor. In order to obtain a bound on the complexity of this problem, we propose to use classical shadows with random Pauli measurements \cite{hkp20} for this task. We refer to this method as classical shadows. In practice, the technique of locally-biased classical shadows \cite{hadfield20} will be significantly more efficient. And this claim continues to hold in the context of weighted MaxCut. See \cite[Section 2]{hadfield20} for a concise definition of classical shadows.

Let $S$ be the number of state preparations of $\rho$, and let $\nu^{(s)}$ for $s\in[S]$ be a random variable obtained from the classical shadows measurement procedure which, in an unbiased fashion, estimates $\tr(H\rho)$. That is, $\E(\nu^{(s)})=\tr(H \rho)$. Set $\nu=\frac1S\sum_{s\in[S]}\nu^{(s)}$. 
In a similar fashion, let $\mu_e^{(s)}$ be a random variable associated with estimating $\tr(O_e \rho)$ for $e\in E$ using classical shadows, and set $\mu_e = \frac1S\sum_{s\in[S]} \mu_e^{(s)}$.
We consider the scaling of $S$ with respect to the size of the graph $G$ such that:
\begin{itemize}
    \item $\nu$ estimates $\tr(H \rho)$ to additive error $\varepsilon$;
    \item $\nu$ estimates $\tr(H \rho)$ to additive error $\varepsilon | \tr( H \cdot F(m^*) )|$.
\end{itemize}
The second scaling should be understood in the context of the VQE algorithm when $\rho$ is close to $\rho_\rel$, since we are using the approximation ratio as a success metric. The typical context in which VQE algorithms are used is quantum chemistry. There, one is typically interested in reaching chemical accuracy on the ground state, which is independent of system size. Here, we are instead concerned with recovering an approximate solution with a large approximation ratio, defined relative to a configuration $m^*$ maximising Eq.~\eqref{eqn:cut}. The second bullet point above is therefore estimating $\tr(H\rho)$ to \emph{multiplicative} error $\varepsilon$. And it is this multiplicative-error regime which must scale favorably in the limit of large graphs. Both of these scalings require the following
\begin{lemma}\label{lem:shadows_var_weight2}
Let $\kappa,\delta>0$. Estimating the expectation observables $O_e$ using classical shadows such that, for all $e\in E$,
\begin{equation}
    \mathbb{P} 
    \left(
        | \tr( O_e \rho) - \mu_e | < \kappa 
    \right)
    >
    1 - \delta
\end{equation}
can be accomplished if $S> \frac{2\cdot 3^2}{\kappa^2} \log(2|E|/\delta)$.
\end{lemma}
\begin{proof}
The variance of $\mu_e$ can be found in \cite[Section 2]{hadfield20}. However it suffices to note that $O_e$ is a weight-2 Pauli operator, and the random variable $\mu_e^{(s)}$ is a random variable bounded by the value $3^2$. Hoeffding's inequality then implies that for a fixed edge $e$, we have
\begin{equation}\label{eqn:hoeffding_inside_proof}
    \mathbb{P} 
    \left(
        | \tr( O_e \rho) - \mu_e | > \kappa 
    \right)
    <
    2 \exp \left(
        \frac{ - S \kappa^2 }{2 \cdot 3^2 }
    \right)
\end{equation}
The claim of the lemma follows upon applying a union bound in order to ensure all edges are accurately measured.
\end{proof}

The additive error scaling is contained in the following lemma. (We we do not prove it in detail as it can be easily obtained from adjusting the proof of the lemma associated with the multiplicative scaling.)
\begin{lemma}
    Let $\varepsilon, \delta >0$. With probability $1-\delta$, an $\varepsilon$-accurate (in the additive sense) estimation of $\tr(H \rho)$ is obtained using the classical shadows estimator $\nu$ provided that $S > \frac{3^4}{2\varepsilon^2} |E|^2 \log (2 |E| / \delta )$.
\end{lemma}
\begin{proof}
The proof of Lemma~\ref{lem:shadows_multiplicative_requirements} needs to be very slightly adjusted. Specifically, the increased number of measurements implies that the accuracy of each estimator in Eq.~\eqref{eqn:accuracy_each_pauli} becomes $2\varepsilon/3|E|$. The triangle inequality, applied to $|\tr(H \rho) - \nu|$, then provides additive-error accuracy of $\varepsilon$.
\end{proof}
It is much more important to understand the scaling (with respect to the size of the graph) of the multiplicative-error estimation of the energy.

\begin{lemma}\label{lem:shadows_multiplicative_requirements}
    Let $\varepsilon, \delta >0$. The classical shadows estimator $\nu$ for $\tr(H \rho)$ achieves
    \begin{equation}
        \mathbb{P} 
        \left(
            \left| 
                \frac{ \tr( H \rho) - \nu }
                    { \tr( H \cdot F(m^*)) }
            \right| 
            < \varepsilon
        \right)
        >
        1 - \delta
    \end{equation}
    provided that $S > \frac{2\cdot 3^4}{\varepsilon^2} \log ( 2 |E| / \delta )$.
\end{lemma}
\begin{proof}
Lemma~\ref{lem:shadows_var_weight2} and the requirement on $S$ imply that the estimators $\mu_e$ for $\tr(O_e\rho)$ satisfy, for all $e\in E$,
\begin{equation}\label{eqn:accuracy_each_pauli}
    \mathbb{P} 
    \left(
        \left| 
            \tr( O_e \rho) - \mu_e | < \frac{\varepsilon}{3} 
        \right|
    \right)
    >
    1 - \delta.
\end{equation}
Therefore, with probability greater than $1-\delta$, we have
\begin{equation}\label{eqn:intermediate_high_prob_bound}
    \frac{3}{|E|} \sum_{e\in E} 
    \left|
    \tr( O_e \rho )  - \mu_e 
    \right|
    < \varepsilon.
\end{equation}
Leaving this argument for the moment, consider $\tr(H \cdot F(m^*))$. We have the inequality
$\tr(H \cdot F(m^*)) = \cut(m^*) > |E|/2$ which is a classic result obtained by considering the MaxCut algorithm which randomly assigns the values of the binary variables associated with the vertices. Therefore
\begin{equation}\label{eqn:bound_max_energy}
    \frac{1}{\tr(H \cdot F(m^*))} \le \frac{2}{|E|}.
\end{equation}
The triangle inequality applied to the Hamiltonian in Eq.~\eqref{eqn:H} gives
$
| \tr(H \rho) - \nu | \le \frac32 \sum_{e\in E} | \tr(O_e \rho) - \mu_e |
$
which, combined with the previous display equation, provides
\begin{equation}
    \left| 
        \frac{ \tr( H \rho) - \nu }
            { \tr( H \rho_\rel) }
    \right|
    \le
    \frac{3}{|E|}
    \sum_{e\in E} | \tr(O_e \rho) - \mu_e |.
\end{equation}
Combining this inequality with Eq.~\eqref{eqn:intermediate_high_prob_bound} (which holds with probability greater than $1-\delta$) establishes the lemma.
\end{proof}
Note that the proof provides a result stronger than the announced lemma. Specifically, consider the bound in Eq.~\eqref{eqn:bound_max_energy}. This bound is true, in expectation, for $\tr(H\cdot F(m))$ provided $m$ is chosen uniformly at random. The multiplicative error, and hence number of measurements required, can therefore be understood for all stages of a VQE routine provided that the routine begins in a randomly-chosen encoded state $F(m)$ whence subsequent intermediate states $\rho$ ought have energy also above $|E|/2$.

\section{Recovery of cut value given embedded state}\label{app:recovery_cut_given_rho}

Suppose we have a state $\rho$ which is guaranteed to be an encoded state for some $m$. That is, $\rho=F(m)$; however, we have no direct access to $m$. Our goal is to understand how many copies of $\rho$ we require in order to accurately estimate $\cut(m)$.
We maintain the notation of Supplementary Information~\ref{app:variance_classical_shadows}. That is, we have estimators $\mu_e$ using the classical shadows estimation routine such that $\E(\mu_e) = \tr(O_e \rho)$ for all weight-2 Pauli operators $O_e$ present in the Hamiltonian $H$.

Hoeffding's inequality allows us to understand the probability that $\tr(O_e \rho)$ and $\mu_{e}$ differ by $\kappa$ additive error, given in Eq.~\eqref{eqn:hoeffding_inside_proof}. Since we assume that $\rho=F(m)$ for some unknown configuration $m$, we are guaranteed that $\tr(O_e\rho)=\pm \frac1{\sqrt3}\frac1{\sqrt3}$ hence we need to estimate $\tr(O_e \rho)$ to within $\kappa=\frac13$ additive error. This accuracy will then allow us to determine the parity of the binary variables associated with the vertices of the edge $e$. Moreover, we want to have a $1-\delta$ success-rate that \emph{all} Pauli operators $O_e$ provide the correct parities. A union bound solves this problem. It follows that the condition on the number of measurements $S$, for a $1-\delta$ success-rate of correctly estimating $\cut(m)$ is
\begin{equation}\label{eqn:shots_lower_bound}
    S 
    >
    {2\cdot 3^{ 4 }}
    \cdot
    \log\left( \frac{2|E|}{\delta} \right).
\end{equation}
The condition in Eq.~\eqref{eqn:shots_lower_bound} establishes a lower bound on the number of measurements required.

\section{Alternative relaxations}\label{app:deformation_family}

This section explains how to see our algorithm as providing a family of deformations of algorithms. One instance of this family of deformations coincides with the well-known quantum approximate optimization algorithm \cite{farhi2014qaoa}. Although it is possible to go directly from our algorithm to QAOA, it is attractive to provide an intermediary stage of the deformation. To this end, we introduce a second family of quantum random access codes. We then recall and unify some notation. Finally, we explain the deformation.

Consider the encoding of two binary variables $m=\{m_i\}_{i\in[2]}$ into a qubit $\D(\C^2)$
\begin{equation}
    f^{(2)}(m) 
    = 
    \frac12\left( I + \frac{1}{\sqrt2}\left(m_1 X + m_2 Z\right) \right)
\end{equation}
and denote by $\pi_P^\pm$ the projections onto the $\pm1$ eigenspaces of the Pauli operators $P\in\{X,Y,Z\}$. Measuring in the $X$, $Z$ bases allows us to, respectively, recover the variables $m_1$, $m_2$. For example, $\tr(\pi_X^+ \cdot f^{(2)}(+1, m_2) = \frac12 + \frac1{2\sqrt2}$. This provides the $(2,1,p^{(2)})$-QRAC encoding where $p^{(2)}= \frac12 + \frac1{2\sqrt2}$~\cite{ANTV02,Hayashi2006}. For quantum rounding, set $\xi_1^\pm = \frac12(I\pm\frac1{\sqrt2}(X+Z))$ and set $\xi_2^\pm=X\xi_1^\pm X$. Then $\xi_1^\pm$ are the eigenstates associated with the operator $\frac1{\sqrt2}(X+Z)$ while $\xi_2^\pm$ are associated with $\frac1{\sqrt2}(X-Z)$. See Fig~\ref{fig:three_encodings}b. 

For notational consistency, we shall write $f^{(3)}$ for the map in Eq.~\eqref{eqn:f_for_F}. Let us also write, for the encoding of a single variable $m=\{m_i\}_{i\in[1]}$ into a qubit $\D(\C^2)$, the map $f^{(1)}(m) = \frac12(I+m_1 Z)$. Note that $\tr(\pi_Z^+ \cdot f^{(1)}(+1)) = 1$ whence $f^{(1)}$ may be seen as a $(1,1,1)$-QRAC.

We unify the notation for three measurement procedures. Let $\M^{(1)}$ be the quantum channel on a single qubit associated with measuring in the $Z$-basis. Let $\M^{(2)}$ be the quantum channel on a single qubit associated with uniformly at random selecting a measurement basis $\{\xi_i^+,\xi_i^-\}_{i\in[2]}$ and subsequently measuring the qubit in the chosen basis. Let $\M^{(3)}$ be the quantum channel given in Supplementary Information~\ref{app:measurement} associated with randomly measuring in magic bases.

Let $G=(V,E)$ be a graph and let $V=\cup_{c\in[C]}V_c$ be a partitioning, or coloring, of the graph into $C$ partitions such that any two vertices which share an edge are not contained in the same partition. Fix an instance of the deformation $d\in\{1,2,3\}$. To each partition, $V_c$, associate $n_c$ qubits where 
$n_c=\left\lceil |V_c| / d \right\rceil$. 

The vertices in $V_c$ are then each associated with a unique weight-1 Pauli operator supported on the $n_c$ qubits. For the deformation $d\in\{1,2,3\}$, this Pauli operator is one of the $d\cdot n_c$ possible operators which are respectively tensor products of the single-qubit Pauli operators $\{I,Z\}$, $\{I,X,Z\}$, $\{I,X,Y,Z\}$. We write $P_{c,i}$ for the Pauli operator associated with a vertex $v_{c,i}$, where $i\in[|V_c|]$. To each edge $e$ with associated vertices $v_{c,i},v_{c',i'}$, we declare the weight-2 Pauli operator $O_e = P_{c,i}P_{c',i'}$ and we define the Hamiltonian:
\begin{equation}
    H^{(d)}
    =
    \sum_{e\in E}
    \frac12
    \left(
        I 
        - d\cdot
        O_e
    \right).
\end{equation}
We leave as an exercise the corresponding embeddings $F^{(d)}$ such that we obtain a commutative diagram providing $\tr(H^{(d)}\cdot F^{(d)}(m))=\cut(m)$ for all $m\in\{-1,1\}^{|V|}$.

We continue to consider a fixed deformation $d\in\{1,2,3\}$ and consider the magic state rounding procedure. Let $\rho_g^{(d)}$ denote a maximum-energy eigenstate of the Hamiltonian $H^{(d)}$. The quantum rounding procedure, when supplied with $\rho_g^{(d)}$ applies $\M^{(d)}$ on all qubits. See Figure~\ref{fig:three_encodings}. Let $\gamma^{(d)}$ denote the approximation ratio. We consider the three cases separately for clarity:
\begin{itemize}
    \item The deformation $d=1$ has a Hamiltonian which is diagonal in the computational basis whence the ground state is also diagonal in the computational basis. In fact, this setup is analogous to QAOA and $\rho_g^{(1)}$ is a linear combination of states $F^{(1)}(m^*)$ for all $m^*$ solving the MaxCut problem. The procedure $\M^{(1)}$ does not cause a degradation in the energy of the rounded state. Rather, it simply collapses $\rho_g^{(1)}$ into one instance $F^{(1)}(m^*)$. We therefore find that $\gamma^{(1)}=1$. 
    \item Consider the map $\M^{(2)}$ on a single qubit. In expectation, this map can be expressed as the linear operation $\mathcal{E}^{(2)}$ such that 
    $\mathcal{E}^{(2)}(I)=I$, 
    $\mathcal{E}^{(2)}(Y)=0$,
    and
    $\mathcal{E}^{(2)}(P)=\frac12P$
    for $P\in\{X,Z\}$.
    Proceeding in precisely the same was as that given in Supplementary Information~\ref{app:relaxation}, we obtain an inequality of the form Eq.~\eqref{eqn:one_over_nine_linearity} however the prefactor becomes $1/4$ rather than $1/9$. The end result is an approximation ratio which, in expectation, is bounded by $5/8$. That is, $\E(\gamma^{(2)})\ge 5/8$.
    
    \item The deformation $d=3$ provides the bound in expectation of the approximation ratio $\E(\gamma^{(3)})\ge5/9$.
\end{itemize}

\begin{figure}[ht]
    \centering
    \includegraphics[width=0.9\linewidth]{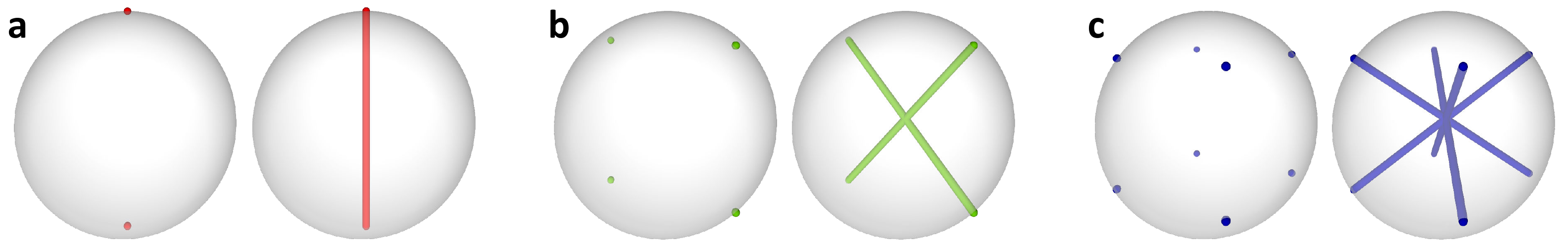}
    \caption{Three different encodings, the states and the measurement bases, of variables into a single qubit. \textbf{(a)} One variable per qubit. \textbf{(b)} Two variables per qubit. \textbf{(c)} Three variables per qubit.}
    \label{fig:three_encodings}
\end{figure}

\begin{remark}
The partition of the vertices is superfluous for the deformation $d=1$. Indeed, since we associate each vertex with an individual qubit, we are guaranteed that the two weight-1 Pauli operators $P_{c,i}$, $P_{c',i'}$, used in the construction of the weight-2 Pauli operator $O_e$, are supported on different qubits.
\end{remark}
\begin{remark}
This deformation may be written in an even more granular fashion. Individual partitions of the vertices may be associated with different deformation parameters $d$. Also different encodings can be used for the same partition. For example, if a color leads to a collection $V_c$ of vertices such that $V_c=3a+b$ with $b\in\{1,2\}$, then we could encode $3a$ vertices into $a$ qubits, using $f^{(3)}$, and the remaining $b$ vertex/vertices could be encoded into a single qubit using the map $f^{(b)}$.
\end{remark}

\begin{remark}
In the limit of graphs with large degree, such as in the Erdős–Rényi model of graphs with fixed probability of an edge between two vertices, our main algorithm will require a coloring with many colors. As the number of vertices associated with each color becomes small, say two vertices per color, or eventually one vertex per color, it is reasonable to use the encoding provided by the $(2,1,0.85)$-QRAC, and eventually the QAOA encoding of one variable into one qubit.
\end{remark}

\section{Weighted MaxCut for an application of aerospace design}\label{app:ply}

In this section we explain how a design problem relevant for the aerospace industry is mapped into a weighted MaxCut graph, which is then addressed by the algorithms proposed. 
In aerospace applications, the use of ply composites offers a significant strength-to-weight advantage for structural components, and have been successfully realized in airframe construction including fuselage, wings, empennage, and flight control surfaces of recent aircraft. Due to the large number of ordering possibilities and complicated nature of the design rules disallowing certain permutations, the optimized arrangement and ordering of the fiber angles comprising the composite stacking sequence poses a computationally challenging combinatorial problem.

Ply composites are composed of plies with varying fiber orientations, which offer a total behavior greater than the sum of the individual plies. The physical and structural properties of the plies depend on the relative fractions of the orientation angles comprising the composite, and therefore the ply percentages (defined as the fractions of each fiber orientation relative to the total number of plies) are often selected to meet a specific structural design requirement. In addition to the structural properties, the ordering of the sequences of plies, referred to as the stacking sequence, also greatly affects the behavior of the composite. There are particular well-known orderings or permutations of plies that can render the composite structurally unsound or bestow it with undesirable properties, such as delaminating or warping. The set of undesirable permutations may be represented by a set of design rules, which specify disallowed sets or orderings that may occur through the stacking sequence. The problem of optimizing ply composites is a challenging problem partly due to the nature and intricacy of these various stacking sequence rules that must be obeyed to ensure a high-quality composite structure.

While specific design rules depend on the composite structure, several common rules are well known~\cite{liu2011bilevel,irisarri2009multiobjective,niu1992composite}. The design rules in Tab.~\ref{table:design-rules} represent some typical examples applied to the set of traditional laminates (consisting of $0$, $\pm45$, and $90$ degree fiber orientations). It is desirable that any composite part be manufactured in accordance to these rules to ensure that structural properties are met and that a high-quality part is produced.

\begin{table}[H]
\centering
\begin{tabular}{c | l} 
Rule & ~Description \\ [0.5ex] 
\hline\hline
 1 & ~Stacking sequence should be balanced (equal number of $+45$ and $-45$ plies) \\ 
 2   & ~Stacking sequence should be symmetric about the midplane, if possible \\
 3 & ~Number of plies in any orientation placed sequentially shall be limited to a maximum (typically four)  \\
 4  & ~A 90-degree change of angle between two adjacent plies is to be avoided
\end{tabular}
\caption{Examples of typical design rules applied to a set of traditional laminates consisting of discrete fiber orientations.}
\label{table:design-rules}
\end{table}

A complete demonstration of a solution to a full stacking sequence optimization problem is beyond the scope of the present work. However, a demonstration of the effectiveness of the quantum method for aerospace problems such as stacking sequence optimization can be achieved by using a simplified yet practically motivated model. In this approach, we determine ordered sets of composite angle sequences of a fixed number of total plies that best minimize the total number of stacking sequence design rule violations.

\begin{figure}[t]
    \centering
    \includegraphics[width=0.85\linewidth]{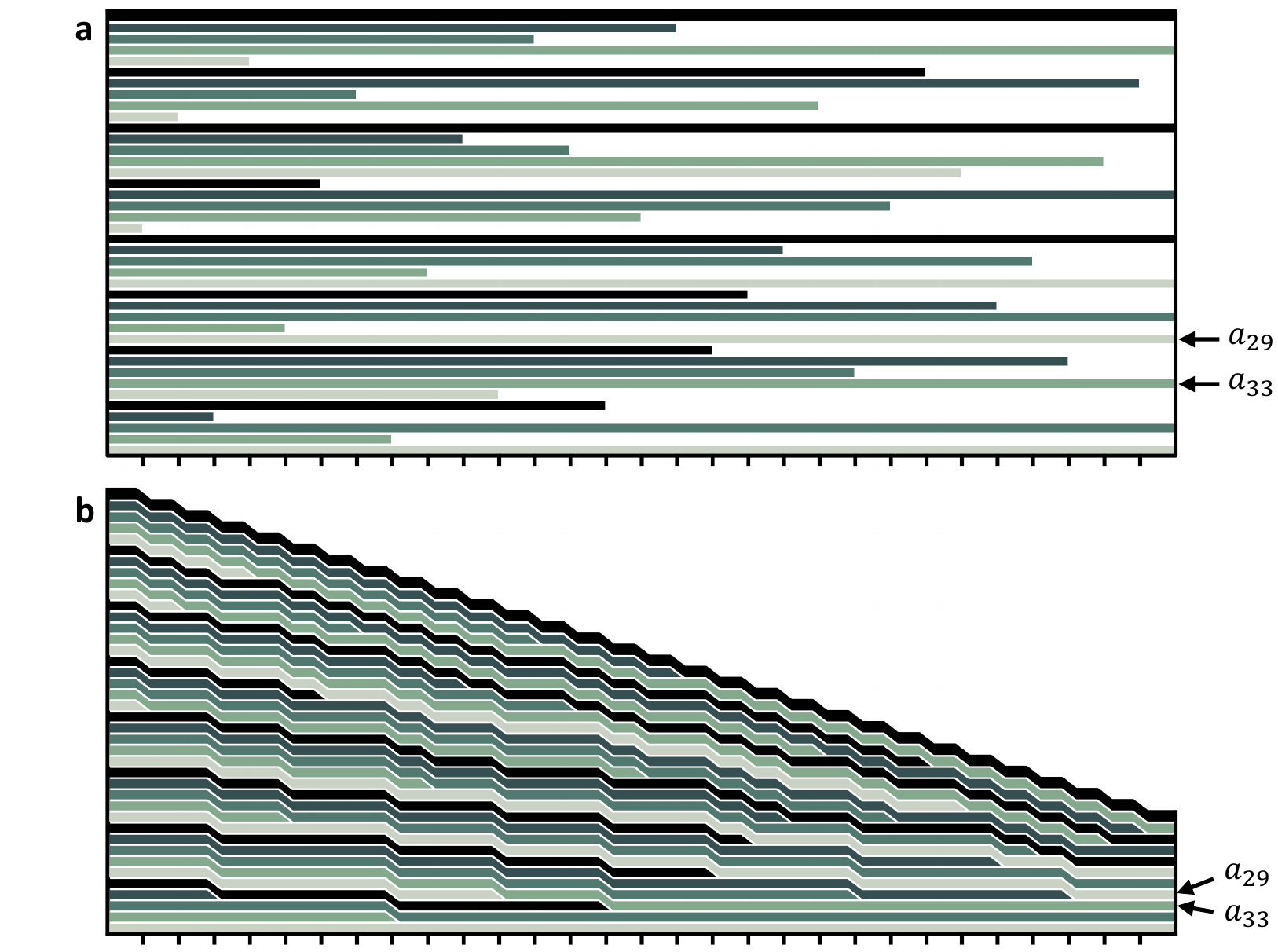}
    \caption{A composite structure shown in profile with $N=40$ plies. Each row corresponds to a ply with fibers oriented at one of four angles $a_i \in \{0^\circ, \pm45^\circ, 90^\circ\}$ relative to an external reference frame. Finding a sequence of fiber angles adhering to design rules is a challenging computational problem. \textbf{(a)} Plies that do not span the full length of the structure are given placeholders, shown in white, that fill the remaining length of the row. \textbf{(b)} With the placeholders removed, plies drop down and the height of the structure is tapered towards the rightmost edge. Plies that were not vertically adjacent in \textbf{(a)} can become partially adjacent in \textbf{(b)}, leading to non-local constraints in the mathematical formulation of the problem.}
    \label{fig:ply-structure}
\end{figure}

Let's consider a composite structure with $N=40$ plies, shown in Fig.~\ref{fig:ply-structure}a in profile with colored rows. Each row $i$ is associated with ply fibers oriented at one of four possible angles of rotation in the plane of the ply, $a_i \in \{0^\circ, \pm45^\circ, 90^\circ\}$, relative to an external reference frame. Note that the plies in Fig.~\ref{fig:ply-structure} are colored in a repeating pattern to guide the eye and not necessarily to indicate a repeating pattern of fiber angles. The design task for the composite structure is to find a good sequence of angles $a_0, a_1, \ldots, a_{N-1}$ that adheres to a set of design rules. Here, we consider design rules that specify angle changes between pairs of plies that should be avoided. Given four angles, there are 16 different two-angle sequences, from which we choose the set of 8 disallowed sequences:
\begin{equation}\label{eq:disallowed}
    \mathcal{D} = \{(0^\circ, 90^\circ), (90^\circ, 0^\circ), (\mp45^\circ, \pm45^\circ), (0^\circ, 0^\circ), (90^\circ, 90^\circ), (\pm45^\circ, \pm45^\circ)\}.
\end{equation}
The first three entries in this set correspond to the well-known design rules restricting ply fiber orientations to no more than $45^\circ$ transitions between two adjacent plies (cf. Rule 4 in Tab.~\ref{table:design-rules}). The latter three entries can be interpreted as restrictions on the number of consecutive plies of identical orientation. Although, in general, design rules may allow up to four consecutive plies of the same orientation (Rule 3 in Tab.~\ref{table:design-rules}), the last three entries here permit a maximum of only one, which can encourage uniform distribution of the angles in a smaller stack. The remaining set of 8 two-angle sequences $\{(0^\circ, \pm45^\circ), (\pm45^\circ, 0^\circ), (\pm45^\circ, 90^\circ), (90^\circ, \pm45^\circ)\}$ are permitted transitions between ply pairs. 

In general, we consider disallowed sequences between ply pairs that may not be vertically adjacent in a given stack. However, these non-adjacent plies are assumed to be connected by a set of contiguous blanks, or placeholders, shown in white in Fig.~\ref{fig:ply-structure}a. Placeholders represent the absence of a physical ply in a stack and can be used, for example, to account for tapering of the composite structure along one direction as illustrated in Fig.~\ref{fig:ply-structure}b. As a result, plies that are non-adjacent in Fig.~\ref{fig:ply-structure}a become adjacent in Fig.~\ref{fig:ply-structure}b once the placeholders are removed and plies drop down. Although traditionally the placeholder blank locations are treated as additional variables to determine the best ply drop locations, here we use a simplified representation that considers a fixed, predetermined set of placeholder locations that are held constant. Future work could explore problem formulations that introduce placeholder locations as variables. In the mathematical formulation we'll discuss shortly, placeholders lead to non-local constraints between decision variables, which may increase the difficulty of finding a good solution. 

The ply composite design challenge can be formulated as a discrete optimization problem to which we apply Algorithm~\ref{alg:maxcut_with_magic} described in Supplementary Information~\ref{app:algorithm}. We start by representing the angle $a_i$ of the $i$th ply by a pair of binary variables $y_i, x_i \in \{0,1\}$ shown in Tab.~\ref{table:binary-variables}. It is sufficient to model the full set of disallowed angle sequences in Eq.~\eqref{eq:disallowed} by comparing the value of the second bit $x_i$ and $x_j$ of the pair of plies $i$ and $j$. For example, $x_i=x_j=1$ indicates the disallowed sequences $(\mp45^\circ, \pm45^\circ)$ and $(\pm45^\circ, \pm45^\circ)$, while $x_i=x_j=0$ captures the remaining sequences in $\mathcal{D}$. The total number of binary variables $x_i$ needed to model the full problem is therefore $N$, one for each ply. Modeling constraints in this way means that the first bits $y_i$ are not needed in our formulation, implying there may be multiple angle assignments for the full structure that obey the constraints. Our formulation could be moved closer still to practical industrial applications by, for example, targeting specific angle percentages for the solution stack or including additional design rules. In particular, design rules such as balancing of the $\pm45^\circ$ angles in the stack (Rule 1 in Tab.~\ref{table:design-rules}), are not addressed here but would be of interest to explore in future work and may serve to further narrow the space of unique solutions.

\begin{table}[h!]
\centering
\begin{tabular}{r | c} 
 $a_i$ & $y_i x_i$ \\ [0.5ex] 
 \hline\hline
 $0^\circ$ & 00 \\ 
 $45^\circ$   & 01 \\
 $90^\circ$ & 10  \\
 $-45^\circ$  & 11
\end{tabular}
\caption{Four possible fiber angles $a_i$ of the $i$th ply encoded as a pair of binary variables $y_i x_i$.}
\label{table:binary-variables}
\end{table}

Mathematically, the disallowed sequences in Eq.~\eqref{eq:disallowed} can be represented by linear equality constraints on the binary variables $x_i, x_j$, which can be converted to a quadratic term $C_{ij}$ of the form,
\begin{equation}
    x_i = x_j \longrightarrow C_{ij} = (x_i - x_j)^2 = x_i + x_j - 2 x_i x_j, \qquad C_{ij} \in \{0,1\}.
\end{equation}
The presence of disallowed sequences leads to $C_{ij}=0$, while the allowed sequences yield $C_{ij}=1$. We assume we are given a set of unique constraints $\mathcal{C} = \{ C_{ij} \}$, which is equivalent to specifying a graph $G(E,V)$ with edges $E$ and vertices $V$. We build the corresponding cost function $F(x)$, with $x$ the set of decision variables $x_i$ appearing in $\mathcal{C}$, by summing over the set of constraints:
\begin{equation}
    F(x) = \sum_{i,j \in \mathcal{C}} \omega_{ij} C_{ij} = \sum_{i,j \in \mathcal{C}} w_{ij} (x_i - x_j)^2  = \sum_{i,j \in \mathcal{C}} \omega_{ij} (x_i + x_j - 2 x_i x_j).
\end{equation}
Here, $\omega_{ij}$ allows for weighting each constraint $C_{ij}$. We then look for angle configurations minimizing the number of design rule violations by maximizing the cost function over $x$:
\begin{equation}
    \max_x F(x).
\end{equation}
Transforming $x_i \rightarrow (1-z_i)/2$, where $z_i \in \{-1,1\}$, we see this problem is equivalent to a weighted MaxCut:
\begin{equation}
    \max_z F(z) =  \max_z \frac12 \sum_{(i,j) \in E} \omega_{ij} (1-z_i z_j),
\end{equation}
where $E$ is the set of edges in the graph associated with the set of constraints $\mathcal{C}$. Algorithm~\ref{alg:maxcut_with_magic} in Supplementary Information~\ref{app:algorithm} can be directly applied to find optimized angle sequences.

The problem instance we address here is generated by a set of local and non-local constraints for a stack of $N=40$ plies. Each vertically adjacent pair of plies in the stack is given a constraint, leading to a linear chain of $N-1=39$ constraints: $\{C_{i,i+1}\}_{i=0}^{N-2}$. In addition, all pairs of plies that are not vertically adjacent in Fig.~\ref{fig:ply-structure}a, but are separated by a contiguous set of placeholders, are given non-local constraints $C_{i,j>i+1}$. See, for example, plies 29 and 33 in Fig.~\ref{fig:ply-structure}, which become adjacent once the placeholders are removed. We generate 29 such non-local constraints according to the placeholder locations shown in Fig.~\ref{fig:ply-structure}a. For each constraint, local or non-local, we add an integer weight $\omega_{ij}$ equal to the length of the interface between the corresponding plies $i,j$, as measured by the tick marks on the lower axes of Fig.~\ref{fig:ply-structure}. For example, plies 29 and 33 have an interface of length 3 in Fig.~\ref{fig:ply-structure}b and, therefore, constraint $C_{29,33}$ is weighted by an amount $\omega_{29,33}=3$. The full adjacency list is shown in Eq.~\eqref{eq:adjacency-ply}. 

The maximum possible cost function value for this problem instance is 641. This solution is illustrated in Fig.~\ref{fig:ply-structure-solutions}b,d for the planar graph and ply composite structure, respectively. The $i$th node of the graph represents the $i$th ply in the stack. Edge thickness is proportional to the constraint weight and edge color denotes satisfied (black) or unsatisfied (blue) constraints. The corresponding ply composite structure is shown underneath with unsatisfied constraints indicated by blue lines at the interfaces of the corresponding plies. We note that a trivial solution is found by alternating the ply bit values $x_i$ (e.g., $x=0101...01$), which satisfies all of the local constraints and reaches a cost function value of $F(0101\ldots01) = 503$. 

Applying Algorithm~\ref{alg:maxcut_with_magic} to this problem instance reduces the number of qubits needed from 40 to 15 (a compression factor of 2.66 compared to a maximum possible of 3.0) and yields a best obtained solution of 617 on quantum hardware. This solution is illustrated in Fig.~\ref{fig:ply-structure-solutions}a,b (and also in Fig.~\ref{Figure3}d). See Secs.~\ref{app:numerics} and~\ref{app:experiments} for more information about the numerics and experiments.

\begin{figure}[t]
    \centering
    \includegraphics[width=0.85\linewidth]{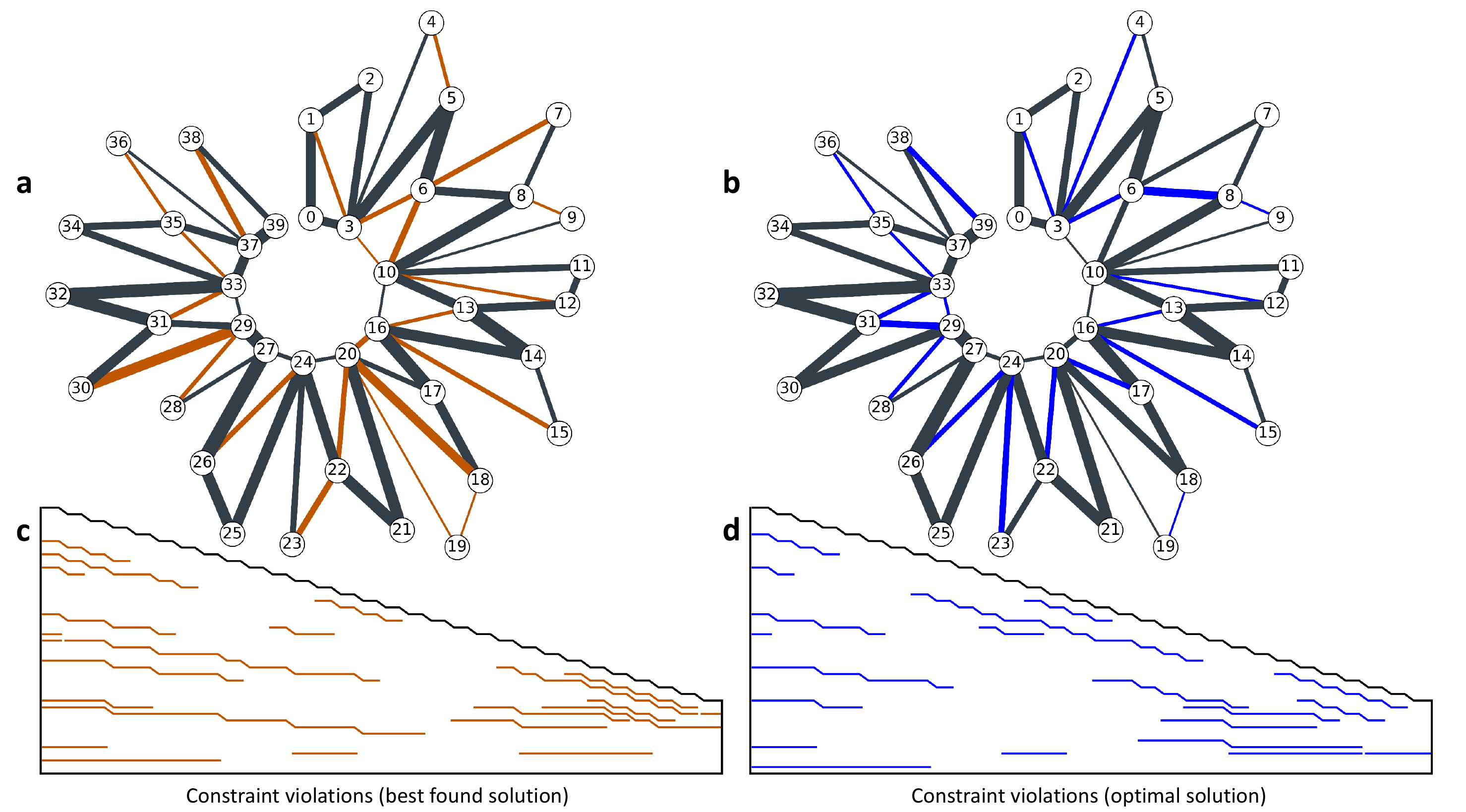}
    \caption{Solutions for the ply composite problem, a 40-node weighted MaxCut instance. Planar graphs for \textbf{(a)} the solution from a hardware experiment and \textbf{(b)} the optimal solution. Each node corresponds to a ply, with the number indicating position in the stack. Edge thickness is proportional to the constraint weight and edge color denotes satisfied (black) or unsatisfied (orange or blue) constraints. Corresponding locations of unsatisfied constraints in the ply composite structure for both solutions are shown in \textbf{(c)} and \textbf{(d)} with colored lines at the interface of the associated ply pairs.}
    \label{fig:ply-structure-solutions}
\end{figure}

\section{Numerical experiments}\label{app:numerics}

All MaxCut problem instances used both in simulation and on hardware were modeled using the \textsf{docplex} interface provided by Qiskit Optimization~\cite{aleksandrowicz2019qiskit}. These \textsf{docplex} models were then solved classically using IBM ILOG CPLEX~\cite{cplex}, then converted into a non-diagonal Hamiltonian, $H$ as given in Eq.~\ref{eqn:H} and represented internally as a Qiskit \textsf{opflow} object. Simulations that use the exact relaxed solution, such as those used to generate Fig~\ref{Figure1}e, were solved using Qiskit's \textsf{NumpyMinimumEigensolver} and were then rounded by directly computing Tr$\left(H \cdot \rho_g\right)$ or by sampling from $\rho_g$ in random magic bases.  Simulations that include solving for an approximate relaxed solution were done using Qiskit's implementation of the Variational Quantum Eigensolver (VQE). For these VQE simulations the COBYLA optimizer was used with \textsf{maxiter=25,000}, \textsf{rhobeg=1.0}, and \textsf{tol=None}. The same hardware-efficient ansatz used during our hardware experiments was used for all VQE simulations. See Supplementary Information~\ref{app:experiments} for further details. 

The adjacency list for all problem instances executed on hardware are included here. 
$G_{16}$ and $G_{40}$ correspond to the 16- and 40-node unweighted graphs used in Figs~\ref{Figure2} and~\ref{Figure3}a. $G^W_{40}$ is the weighted graph corresponding to the 40-node ply composite problem in Fig.~\ref{Figure3}d. The adjacency lists are provided in the form ($m_i$, $m_j$, $w_{ij}$).

\begin{equation}
\begin{split}
    G_{16} = [& (0, 4, 1), (0, 12, 1), (0, 15, 1), (1, 2, 1), (1, 13, 1), (1, 14, 1), (2, 12, 1), (2, 13, 1),\\
    & (3, 5, 1), (3, 7, 1), (3, 11, 1), (4, 9, 1), (4, 14, 1), (5, 8, 1), (5, 10, 1), (6, 10, 1),\\
    & (6, 11, 1), (6, 12, 1), (7, 11, 1), (7, 15, 1), (8, 9, 1), (8, 10, 1), (9, 14, 1), (13, 15, 1)]
\end{split}    
\end{equation}

\begin{equation}
\begin{split}
    G_{40} = [& (0, 15, 1), (0, 18, 1), (0, 34, 1), (1, 9, 1), (1, 15, 1), (1, 35, 1), (2, 3, 1), (2, 5, 1),\\
    & (2, 25, 1), (3, 17, 1), (3, 38, 1), (4, 35, 1), (4, 37, 1), (4, 39, 1), (5, 21, 1), (5, 26, 1),\\
    & (6, 22, 1), (6, 23, 1), (6, 26, 1), (7, 12, 1), (7, 29, 1), (7, 35, 1), (8, 11, 1), (8, 28, 1),\\
    & (8, 33, 1), (9, 20, 1), (9, 24, 1), (10, 21, 1), (10, 27, 1), (10, 39, 1), (11, 21, 1), (11, 31, 1),\\
    & (12, 23, 1), (12, 39, 1), (13, 25, 1), (13, 29, 1), (13, 30, 1), (14, 28, 1), (14, 31, 1), (14, 32, 1),\\
    & (15, 18, 1), (16, 17, 1), (16, 19, 1), (16, 33, 1), (17, 23, 1), (18, 24, 1), (19, 24, 1), (19, 25, 1),\\
    & (20, 27, 1), (20, 32, 1), (22, 27, 1), (22, 37, 1), (26, 36, 1), (28, 34, 1), (29, 36, 1), (30, 34, 1),\\
    & (30, 38, 1), (31, 37, 1), (32, 33, 1), (36, 38, 1)]
\end{split}
\end{equation}

\begin{equation}\label{eq:adjacency-ply}
    \begin{split}
        G^{W}_{40} = [&(0, 1, 16),(1, 2, 12),(2, 3, 12),(3, 4, 4),(4, 5, 4),(5, 6, 23),(6, 7, 7),(7, 8, 7),\\
        & (8, 9, 2),(9, 10, 2),(10, 11, 10),(11, 12, 10),(12, 13, 13),(13, 14, 24),(14, 15, 6),(15, 16, 6),\\
        &(16, 17, 22),(17, 18, 15),(18, 19, 1),(19, 20, 1),(20, 21, 19),(21, 22, 19),(22, 23, 9),(23, 24, 9),\\
        & (24, 25, 18),(25, 26, 18),(26, 27, 25),(27, 28, 5),(28, 29, 5),(29, 30, 17),(30, 31, 17),(31, 32, 21),\\
        & (32, 33, 21),(33, 34, 11),(34, 35, 11),(35, 36, 3),(36, 37, 3),(37, 38, 8),(38, 39, 8),(18, 20, 14),\\
        & (8, 10, 18),(35, 37, 11),(3, 5, 19),(27, 29, 25),(14, 16, 18),(6, 8, 13),(37, 39, 22),(22, 24, 17),\\
        & (10, 12, 3),(33, 35, 3),(1, 3, 4),(10, 13, 15),(33, 37, 16),(17, 20, 7),(0, 3, 14),(29, 31, 10),\\
        & (24, 26, 7),(20, 22, 7),(6, 10, 9),(31, 33, 6),(16, 20, 8),(3, 6, 6),(13, 16, 4),(24, 27, 5),\\
        & (20, 24, 4),(29, 33, 3),(10, 16, 2),(3, 10, 1)]        
    \end{split}
\end{equation}

\section{Experiments on quantum Hardware}\label{app:experiments}

\begin{figure}[t]
    \centering
    \includegraphics[width=0.65\linewidth]{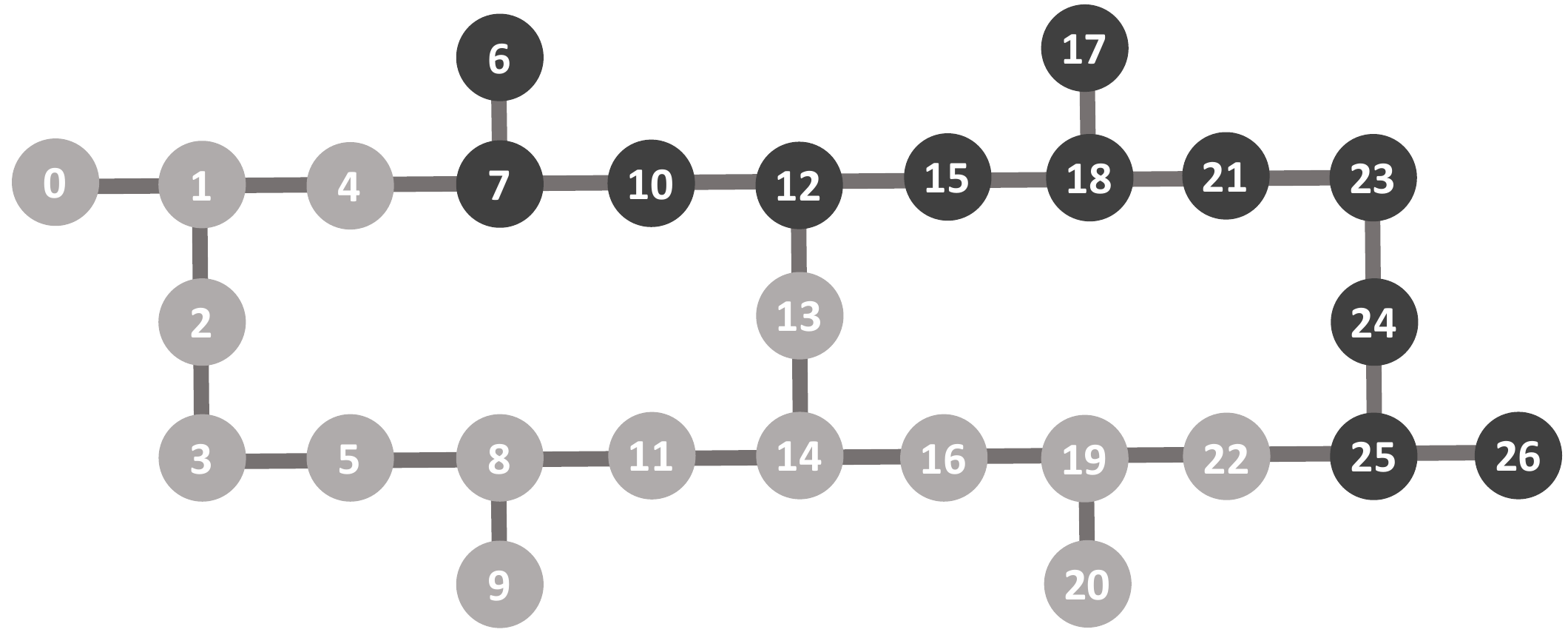}
    \caption{Connectivity of the 27-qubit IBM Quantum device \emph{ibmq\_dublin}.}
    \label{fig:dublin-layout}
\end{figure}

All hardware experiments were run on the 27-qubit device \emph{ibmq\_dublin} (Quantum Volume 64) with connectivity shown in Fig.~\ref{fig:dublin-layout}. The qubits used for a given experiment were chosen according to the device noise. That is, we generated a list $Q_\text{Dublin}$ of the device's qubits such that the first $2\leq k \leq 27$ qubits form a contiguous subset of its connectivity graph (Fig.~\ref{fig:dublin-layout}). Then, the ordering of $Q_\text{Dublin}$ was determined heuristically to prioritize using the least noisy qubits.

\begin{equation}\label{eq:qubit-order-dublin}
    Q_\text{Dublin} =  [8,11,9,5,3,14,13,16,19,20,22,2,1,0,4,7,10,12,15,18,21,25,26,17,6,24,23]
\end{equation}

The variational circuit we use is a depth-$l$ hardware-efficient variational ansatz introduced in Ref.~\cite{kandala2017hardware} with $l$ layers of arbitrary single-qubit rotations interleaved with $l-1$ entangling blocks of controlled-$Z$ gates. Note that a depth-$1$ instance of this ansatz yields a product state, while $l>1$ produces entangled states. For a circuit with $n$ qubits, our ansatz uses the contiguous subset of qubits given by the first $n$ elements of Eq.~\eqref{eq:qubit-order-dublin}. Within each repeated block of the ansatz, the controlled-$Z$ gates are applied once between each contiguous pair of qubits. For example, the 40-node MaxCut instances of Fig.~\ref{Figure3} were run on the qubits indicated with light gray in Fig.~\ref{fig:dublin-layout}. 

VQE experiments run on hardware were optimized using the Simultaneous Perturbation Stochastic Approximation (SPSA) optimizer with 500 iterations. 
We employ single-qubit readout error mitigation by tensoring together measurement filters computed on each qubit. This requires only a linear number of calibration circuits at the cost of not accounting for correlations between qubits. A modified version of Qiskit's \textsf{TensoredMeasFitter} was used to implement this scheme. 

\newpage

\end{document}